\documentclass[12pt,a4paper]{article} 
\usepackage{amsmath,amssymb,amsfonts}
\usepackage[colorlinks]{hyperref}

\usepackage{relsize}

\usepackage{diagbox}

\usepackage{thmtools,thm-restate}

\usepackage{graphics} 

\usepackage{graphicx}

\usepackage{ cmll } 

\usepackage[noxy]{virginialake} 

\usepackage{pifont}%
\newcommand{\cmark}{\ding{51}}%

\usepackage[dvipsnames]{xcolor} 

\usepackage{bussproofs}

\usepackage[bibliography=common]{apxproof}

\usepackage{tikz}
\usepackage{nicematrix}

\theoremstyle{plain} 
\newtheorem{theorem}{Theorem}

\newtheorem{lemma}[theorem]{Lemma}
\newtheorem{corollary}[theorem]{Corollary}

\theoremstyle{definition}
\newtheorem{definition}[theorem]{Definition}
\newtheorem{example}[theorem]{Example}
\newtheorem{remark}[theorem]{Remark}

\newcommand{\st}{*}

\newcommand{\calL}{\mathcal{L}}

\newcommand{\calA}{\mathcal{A}}

\newcommand{\CFL}{\mathsf{CFL_e}}
\newcommand{\CFLw}{\mathsf{CFL_{ew}}}

\newcommand{\BCFLw}{\mathbf{CFL_{ew}}}

\newcommand{\BCFLcw}{\mathbf{CFL_{ecw}}}

\newcommand{\BCFLc}{\mathbf{CFL_{ec}}}
\newcommand{\CLL}{\mathsf{LL}}
\newcommand{\FL}{\mathsf{FL_e}}
\newcommand{\F}{\mathbf{F}}

\newcommand{\rng}{\text{rng}}
\newcommand{\BCFL}{\mathbf{CFL_e}}
\newcommand{\BFL}{\mathbf{FL_e}}
\newcommand{\BCLL}{\mathbf{LL}}
\newcommand{\IPC}{\mathsf{IPC}}
\newcommand{\CPC}{\mathsf{CPC}}

\newcommand{\BIMALL}{\mathbf{IMALL}}
\newcommand{\MALL}{\mathsf{MALL}}
\newcommand{\BMALL}{\mathbf{MALL}}

\newcommand{\BILL}{\mathbf{ILL}}

\newcommand{\BAIMALL}{\mathbf{IMALL_{w}}}

\newcommand{\BELL}{\mathbf{ELL}}
\newcommand{\ALL}{\mathsf{LLW}}
\newcommand{\BALL}{\mathbf{LLW}}

\newcommand{\BRLL}{\mathbf{LLC}}
\newcommand{\AMALL}{\mathsf{MALL_{w}}}
\newcommand{\BAMALL}{\mathbf{MALL_{w}}}

\newcommand{\BRMALL}{\mathbf{MALL_{c}}}

\newcommand{\BRIMALL}{\mathbf{IMALL_c}}

\newcommand{\BAILL}{\mathbf{ILLW}}

\newcommand{\BRILL}{\mathbf{ILLC}}

\newcommand{\BFLw}{\mathbf{FL_{ew}}}

\newcommand{\BFLc}{\mathbf{FL_{ec}}}
\newcommand{\LJ}{\mathbf{LJ}}
\newcommand{\LK}{\mathbf{LK}}
\newcommand{\LKu}{\mathbf{LK_u}}
\newcommand{\LKb}{\mathbf{LK_b}}
\newcommand{\LKajab}{\mathbf{LK_!}}
\newcommand{\LJu}{\mathbf{LJ_u}}
\newcommand{\LJb}{\mathbf{LJ_b}}
\newcommand{\LJajab}{\mathbf{LJ_!}}
\newcommand{\LKminus}{\mathbf{LK^-}}
\newcommand{\LKnn}{\mathbf{LK_{nn}}}
\newcommand{\G}{\mathbf{G}}

\newcommand{\bigast}{\mathop{\scalebox{1.5}{\raisebox{-0.2ex}{$\ast$}}}}%
\newcommand{\bigplus}{\mathop{\scalebox{1.5}{\raisebox{-0.2ex}{+}}}}%

\renewcommand{\phi}{\varphi}

\newcommand{\blue}[1]{\textcolor{blue}{#1}}
\newcommand{\green}[1]{\textcolor{ForestGreen}{#1}}
\newcommand{\rah}[1]{\textcolor{red}{#1}}

\begin{document}

\title{Proof Complexity of Linear Logics}
\author{Amirhossein Akbar Tabatabai\footnote{\texttt{amir.akbar@gmail.com}}, Raheleh Jalali\footnote{\texttt{rjk53@bath.ac.uk}}\\
\small{Bernoulli Institute, University of Groningen$^*$}\\
\small{Department of Computer Science, University of Bath$^\dagger$}
}
\date{}
\maketitle
\begin{abstract}
Proving proof-size lower bounds for $\LK$, the sequent calculus for classical propositional logic, remains one of the major open problems in proof complexity. We shed new light on this challenge by isolating the power of structural rules and showing that their combination is dramatically stronger than any individual structural rule alone, even in the presence of the controlled structural rules provided by linear exponentials.

It is easy to see that $\LK$ without the weakening rule is significantly weaker than $\LK$ with respect to proof complexity. 
It therefore remains to study the impact of eliminating contraction and cut.
Working over the Full Lambek calculus with exchange, $\BFL$, as a base system, 
we begin with the role of contraction.
We construct families of $\BFL$-provable formulas that require exponential-size proofs in affine linear logic $\BALL$, yet admit polynomial-size proofs once contraction is restored. This yields exponential proof-size lower bounds for $\BFL$-provable formulas in $\BALL$, and consequently in $\BMALL$, $\BAMALL$, and full classical linear logic $\BCLL$.
We then investigate the role of cut.
We exhibit sequents with polynomial-size $\BFL$-proofs that nevertheless require exponential-size proofs in cut-free $\LK$. 
This shows that the cut rule alone provides an exponential speed-up over the combination of weakening and contraction.
As a consequence, we obtain exponential separations between several linear calculi and their cut-free counterparts.\\

\end{abstract}
\maketitle
\newpage
\tableofcontents
\newpage

\section{Introduction}
Proof complexity is an interdisciplinary field that uses tools from logic, complexity theory, and combinatorics. The main goal of proof complexity is to understand the complexity of theorem-proving procedures. In other words, it aims to determine how efficient proof systems are when they prove their theorems. This naturally leads to the central problem of finding lower bounds on proof size in proof systems. The difficulty of this problem is what makes it interesting: its deep connection with computational complexity.

The field was motivated by the seminal work of Cook and Reckhow \cite{Cook}, who established a fundamental connection between proof size and the separation of complexity classes. They showed that NP $\neq$ coNP if{f} there are superpolynomial lower bounds for \emph{every} proof system for \emph{classical} propositional logic, $\CPC$. This conjectured lower bound for arbitrary proof systems has never been established. In fact, despite decades of research, even the single case of the Frege system (the standard proof system for $\CPC$, polynomially equivalent to sequent calculus $\LK$ and natural deduction) still resists all known lower-bound techniques.

Facing such a strict barrier, progress has therefore focused on weaker systems. The first natural approach was restricting the form of formulas and deduction rules (sometimes in algebraic and semi-algebraic disguise). This led to a plethora of exponential lower bound results for Resolution \cite{Haken}, the Nullstellensatz system \cite{Null}, Cutting Planes \cite{Cutting1,Cutting2}, Polynomial Calculus \cite{PCalculus1,PCalculus2}, and bounded-depth Frege systems \cite{BDepth1,BDepth2,BDepth4,BDepth3}.

Another successful way of avoiding the $\LK$ barrier is to change classical logic itself, either by weakening it (e.g., forbidding full involution $\neg \neg A \to A$) or by extending the language (e.g., with modalities). As these logics are usually $\mathrm{PSPACE}$-hard, their study is closely connected to another notoriously hard problem in complexity theory, namely $\mathrm{NP}$ vs.\ $\mathrm{PSPACE}$. Since this problem is, in a sense, easier than $\mathrm{NP}$ vs.\ $\mathrm{coNP}$, one may expect stronger results on this front. Indeed, there are many known exponential lower bounds for these systems, including Frege systems for intuitionistic and some modal logics \cite{Hrubes1,Hrubes2,Hrubes}, superintuitionistic and modal logics of infinite branching \cite{jerabek,Jerabek2}, intuitionistic substructural logics \cite{Raheleh}, and propositional default logic \cite{Beyersdorff}. Owing to their expressive power and numerous applications in computer science, the proof complexity of these non-classical systems has received extensive attention in recent years.

\subsection{Our Contribution} 

In this work, we follow the same strategy of weakening classical logic; however, our approach differs conceptually. To illustrate the role of full involution, the backbone of \emph{classical reasoning}, we keep the axiom $(\neg \neg A \to A)$ intact and also allow unrestricted formulas in our proofs. Instead, we weaken the logic by pruning its \emph{structural rules}. By removing contraction and the cut rule one at a time\footnote{In Subsection \ref{subsubsection cont}, we will explain that the case where the weakening rule is eliminated is easy to address.}, sometimes even while allowing controlled structural rules via linear exponentials, one can isolate the precise contribution of each rule to the strength of classical proof systems. 

To state our results concisely, we introduce the following notion. Let $G$, $H$, and $I$ be three sequent calculi (possibly for different logics). We say \emph{$H$ simulates $G$ over $I$}, denoted $G \leq_I H$, if for every sequent $S$ provable in $I$,
\[
\text{if} \qquad G \vdash^\pi S \qquad \text{then} \qquad \exists \pi' \;\text{with}\; |\pi'| \leq |\pi|^{O(1)} \;\text{such that}\; H \vdash^{\pi'} S,
\]
where by $G \vdash^\pi S$ we mean that $\pi$ is a proof of $S$ in $G$, and similarly for $H \vdash^{\pi'} S$.

Now, recall that $\BCFL$, $\BCFLw$, and $\BCFLc$ denote the classical full Lambek calculus with exchange, with exchange and weakening, and with exchange and contraction, respectively; $\BMALL$, $\BAMALL$, and $\BRMALL$ denote the multiplicative–additive fragment of linear logic, with weakening, and with contraction, respectively; $\BCLL$, $\BALL$, and $\BRLL$ denote full classical linear logic, with weakening, and with contraction, respectively; and $\BFL$, $\BIMALL$, and $\BILL$ are the intuitionistic counterparts of $\BCFL$, $\BMALL$, and $\BCLL$, respectively. For a calculus $G$ and any $X \subseteq \{c, w\}$, by $GX$ we simply mean taking the system obtained by adding the rules in $X$ to $G$. And by $G^-$ we mean the calculus $G$ without the cut rule.
Our main results are summarized in Table~\ref{table: summary}.

\subsubsection{Contraction-free Calculi}
We prove an \emph{exponential separation} between $\BALL$ and $\LK$ for $\mathbf{FL_e}$-provable formulas. More precisely, we identify a sequence $\{A_n\}_{n \in \mathbb{N}}$ of $\BFL$-provable formulas such that every $\BALL$-proof of $A_n$ has size\footnote{The size of a formula or a proof is the number of symbols it contains.} exponential in $|A_n|$, whereas $A_n$ has polynomial-size proofs in $\LK$. 
This implies an exponential lower bound on the proof size
of $\mathbf{FL_e}$-provable formulas in $\BALL$ and hence for each of the following contraction-free calculi 
\[
G \in \{\BCFL, \BCFLw, \BMALL, \BAMALL, \BCLL\}.
\]
Note that the separating formulas $A_n$ are $!$-free, since they are written in the language of $\mathbf{FL_e}$. Without this $!$-freeness, separation results are trivial. The reason is the high decision complexity of $\BALL$. In fact, $\BALL$ is $\mathrm{TOWER}$-complete, where $\mathrm{TOWER}$ denotes the class of problems decidable in time bounded by a tower of exponentials \cite{lazic2015nonelementary}. By a simple argument, proof size of $\BALL$-provable formulas cannot be bounded by any fixed-height tower: otherwise, one could enumerate all proofs and decide provability within that bound, contradicting $\mathrm{TOWER}$-hardness. In contrast, every formula has an exponential-size proof in $\LK$, establishing the separation. For $\BCLL$, the situation is even more extreme, since the logic is undecidable \cite{Lincoln}.

For $!$-free formulas, however, $\BALL$ is conservative over $\BAMALL$, so its decision problem remains $\mathrm{PSPACE}$-complete \cite{Horcik}, making the exponential separation interesting. Moreover, one might hope that employing exponentials in proofs could provide short proofs for hard $!$-free formulas. However, our results show that even the full machinery of linear exponentials does not reduce the hardness of these formulas.

\subsubsection{Weakening-free Calculi}\label{subsubsection cont}

Similar to the situation for $\BALL$, since provability in $\mathbf{CFL_{ec}}$ is Ackermann-complete~\cite{urquhart}, separating $\mathbf{CFL_{ec}}$ from $\LK$ would be a trivial task. Even better, since provability in $\BRLL$ for $!$-free formulas is also Ackermann-complete \cite{lazic2015nonelementary}, and $\BRLL$ is conservative over $\mathbf{CFL_{ec}}$, one can find $\mathbf{CFL_{ec}}$-provable formulas having relatively shorter $\LK$-proofs than $\BRLL$-proofs.
In this sense, the proof-complexity role of the weakening rule is already clear, and there is little need to pursue the matter further.

\subsubsection{Cut-free Calculi}
We now turn to the role of the cut rule and reveal its remarkable impact.
It is often believed that the cut rule dramatically increases the power of sequent calculi. We confirm this intuition by constructing a sequence $\{S_n\}_{n \in \mathbb{N}}$ of sequents such that $S_n$ has a polynomial-size proof in the weak system $\BFL$, but requiring exponentially long proofs in $\LK^{-}$, i.e., $\LK$ without cut. 
This provides a clear illustration of the power of cut: even when contraction and weakening are fully available, removing cut alone leads to an exponential blow-up. Consequently, we establish an exponential speed-up between any of the following calculi and their cut-free counterparts:
\[
\BFL X, \BCFL X, \BIMALL X, \BMALL X, \BILL X, \BCLL X, \quad X \subseteq \{c, w\}.
\]

\begin{table*}
\begin{tabular}{|c|c|c|c|c||c|c|c|c|}
\hline

& $\emptyset$ & $\{c\}$ & $\{w\}$ & $\{c,w\}$ &$\{cut\}$ & $\{cut, c\}$ & $\{cut, w\}$ & $\{cut, c, w\}$ \\ \hline
$\{cut\}$ &   \green{\cmark}  & \green{\cmark} & \green{\cmark} & \green{\cmark} & --- & --- & --- & --- \\ \hline
$\{cut, c\}$ &   \green{\cmark}  & \green{\cmark} & \green{\cmark} & \green{\cmark} & ? & --- & ? & --- \\ \hline
$\{cut, w\}$ &   \green{\cmark}  & \green{\cmark} & \green{\cmark} & \green{\cmark} & ? & ? & --- & --- \\ \hline
$\{cut, c, w\}$ &   \green{\cmark}  & \green{\cmark} & \green{\cmark} & \green{\cmark} & \rah{\cmark$_{\mathsmaller{1}}$} & \blue{\cmark$_{\mathsmaller{2}}$} & \rah{\cmark$_{\mathsmaller{1}}$} & ---  \\ \hline
\end{tabular}
\vspace{5pt}
\caption{Summary of results. Let $X, Y \subseteq \{c, w, cut\}$. Rows are indexed by $X$ and columns by $Y$. For any $\G \in \{\mathbf{CFL_e^-}, \mathbf{MALL^-}, \mathbf{CLL^-}\}$, the symbols \green{\cmark} and \rah{\cmark$_{\mathsmaller{1}}$} indicate that $\G X \nleq_{\BFL} \G Y$, as proved in Theorem~\ref{MainCutFree} and Theorem~\ref{Thm: main}. The symbol \blue{\cmark$_{\mathsmaller{2}}$} indicates that $\G X \nleq_{\BFLc} \G Y$, which is known due to the high complexity of the contractive systems (see subsection \ref{subsubsection cont}). The symbol --- (resp.\ ?) means that the problem is trivial since $X \subseteq Y$ (resp.\ still open). Comparisons among cut-free systems are omitted, as all such cases are currently open.}  \label{table: summary}
\end{table*}

\subsection{Methods}
To prove our main theorems as outlined above, we use two different interdisciplinary methods.

\subsubsection{$\BALL$ and Chu's Translation}
Chu's construction is a well-known categorical technique for generating self-dual categories and, consequently, for constructing categorical models of classical linear logic from models of intuitionistic linear logic \cite{Barr,Chu}, \cite[Section 3.3]{Hyland}. In special cases, analogous constructions have also been used in algebraic settings to obtain Girard quantales \cite{Quantale}, involutive residuated lattices \cite{tsinakis2006minimal}, and algebraic models of Nelson's constructive negation from Heyting algebras \cite{odintsov2008constructive}.

In this paper, we adapt Chu's construction to define a translation of linear formulas, referred to as \emph{Chu's translation}. Via this translation, proofs in classical linear systems can be feasibly transformed into proofs in their intuitionistic counterparts without distorting the disjunction connectives in the proofs, where the computational content is often concentrated. Consequently, when combined with circuit-extraction techniques that are naturally suited to the intuitionistic setting, Chu's translation enables the extraction of computational content from classical contraction-free proofs and thereby allows circuit lower bounds to be transferred into proof-size lower bounds for classical systems.

This is a crucial new piece of machinery provided by Chu's translation, as it isolates the contraction rule as a potential barrier for lower bounds for $\mathbf{LK}$. Consequently, it offers a new approach to the proof complexity of $\mathbf{LK}$ via its classical yet substructural fragments, with the aim of overcoming existing barriers by allowing more controlled uses of contraction. In this sense, our separation result should be viewed as one instance of a broader strategy that we propose.

It is also worth mentioning that, although Chu's construction is well known in category theory and has been used in algebraic settings under the name of the \emph{twist product}, to the best of our knowledge it has not previously been employed in the proof theory of substructural or linear logics. Hence, we believe that identifying the proof-theoretic power of Chu's construction and importing it into proof theory as a new form of translation with potential for further applications constitutes a conceptual contribution of our paper, which we consider as important as the concrete separation results it establishes.

\subsubsection{Cut-freeness and Feasible Interpolation}

Feasible interpolation, introduced in \cite{krajivcekFeasible,krajivcekfeasible2}, is one of the primary tools for proving exponential lower bounds for classical proof systems. Much of the literature on modal and superintuitionistic proof complexity also revolves around the feasibility of the disjunction property, a variant of feasible interpolation \cite{tabatabai2025,BussMints,PudlakBuss,Ferrari,jevrabek2006,mints2004,Pudlak}. See \cite{AmirProofComp} for a recent survey.

Our reference point is Kraj\'{i}\v{c}ek's exponential separation between $\LK$ and $\LK^-$ \cite{Krajicek}, which motivates our approach to isolating the power of cut. The key step in that separation is showing that $\LK^{-}$ has the \emph{monotone feasible interpolation property}: given an $\LK^{-}$-proof $\pi$ of a sequent $A(\bar p, \bar r) \Rightarrow B(\bar p, \bar s)$,
where $A(\bar p, \bar r)$ is monotone in $\bar p$, there exists a monotone circuit $C(\bar p)$ of size $|C(\bar p)| \leq |\pi|^{O(1)}$ computing a Craig interpolant for $A(\bar p, \bar r) \rightarrow B(\bar p, \bar s)$. The separation then follows by combining this with known exponential lower bounds in monotone circuit complexity to produce an implication that is easily provable in $\LK$, but has exponentially hard interpolating monotone circuits and hence large $\LK^-$-proofs.

The strategy for our third theorem follows the same pattern, but now we need an implication with a short $\BFL$-proof. We achieve this by simulating classical reasoning in $\LK$ within the substructural setting of $\BFL$ using additional initial sequents. The main challenge is to mimic contraction and weakening with sequents whose formulas contain at most one atomic formula, ensuring that the monotone feasible interpolation property of $\LKminus$ remains applicable.

\section{Preliminaries} \label{sec: preliminaries}

Let $\calL_u=\{0,1, \wedge, \vee, \st, \to\}$ be the  language for unbounded substructural logics, define  $\calL_b=\calL_u \cup \{\top, \bot\}$ for its bounded version, and let $\calL_!= \calL_b \cup \{!\}$ be the language for linear logics. To refer to any of these languages, we use the variable $\calL \in \{\calL_u, \calL_b, \calL_!\}$.
We define $\calL_!$-formulas by the following grammar:
\[
F:= p \mid 0 \mid 1 \mid \top \mid \bot \mid !F \mid F_1 \wedge F_2 \mid F_1 \vee F_2 \mid F_1 \to F_2 \mid F_1 \st F_2.
\]
$\calL_u$-formulas and $\calL_b$-formulas are defined similarly. For any language $\calL$, by the abuse of notation, we denote the set of $\calL$-formulas by $\calL$. Define $\neg A:= A \to 0$, $?A:= \neg ! \neg A$, $A+B:= \neg (\neg A \st \neg B)$, 
for any formulas $A$ and $B$. The connectives $!$
and $?$ are known as \emph{exponentials} and $*$ as \emph{multiplication} or \emph{fusion}.
As we consider various forms of linear logic, we recall how our substructural notation translates into Girard’s \cite{Girard} in Table \ref{Fig: notation}. For more on linear logic, see \cite{Avron}, \cite[Subsection 2.3.4.]{Ono}, \cite{Troelstra}. 
\begin{table}[t!]
\begin{center}
\begin{tabular}{ |c| c| }
\hline
linear logic & our notation \\
\hline
$\multimap$ &  $\to$ \\
\hline
$\&$ &  $\wedge$ \\
\hline
$\oplus$ &  $\vee$ \\
\hline
$\otimes$ &  $*$ \\
\hline
$\parr$ &  $+$ \\
\hline
$(-)^\bot$ &  $\neg$ \\
\hline
$\bot$ &  $0$ \\
\hline
$1$ &  $1$ \\
\hline
$\top$ &  $\top$ \\
\hline
$0$ &  $\bot$ \\
\hline
\end{tabular}
\vspace{5pt}
\caption{{\small Translation between linear logic and substructural notation.}}\label{Fig: notation}
\end{center}
\end{table}
If $\Gamma=\{A_1, \dots, A_n\}$, define $\bigast \Gamma: =A_1 * \dots * A_n$, $\bigplus \Gamma: =A_1 + \dots + A_n$ and set $\bigast \emptyset := 1$ and $\bigplus \emptyset := 0$.
The \emph{set of variables} of a formula $A$, denoted $V(A)$ is defined recursively; $V(c)= \emptyset$, where $c$ is a constant, $V(p)=\{p\}$, $V(!A)=V(A)$, $V(A \circ B)= V(A) \cup V(B)$ for $\circ \in \{\wedge, \vee, *, \to\}$. 

By a \emph{logic} $L$ over the language $\calL$, we mean a set of $\calL$-formulas closed under substitution, modus ponens ($A, A \rightarrow B \in L$ implies $B \in L$), adjunction ($A,B \in L$ implies $A \wedge B \in L$), and only in the case $\calL=\calL_!$ also under necessitation ($A \in L$ implies $!A \in L$). 

A \emph{proof system} for a logic $L$ is a polynomial-time function $P$ such that $\rng(P) = L$ (see \cite{Cook}). Examples of proof systems include Hilbert-style proof systems, natural deduction systems and sequent calculi.
We call $\pi$ a \emph{$P$-proof} of a formula $A$ when $P(\pi)=A$, and denote it by $P \vdash^\pi A$. 
If $P$ and $Q$ are two proof systems for the same logic, we say that $Q$ \emph{simulates} $P$ and write $P \leq Q$ if, for every formula $A \in \mathcal{L}$,
\[
\text{if} \qquad P \vdash^\pi A \qquad \text{then} \qquad \exists \pi' \;\text{with}\; |\pi'| \leq |\pi|^{O(1)} \;\text{such that}\; Q \vdash^{\pi'} A.
\]
The proof systems $P$ and $Q$ are \emph{equivalent} when they simulate each other.
If $P$ and $Q$ are sequent calculi, we define $P \leq Q$ similarly, where $A$ ranges over sequents.


\subsection{Sequent Calculi and Frege Systems}
In this subsection, we recall sequent calculi and Frege systems for substructural and linear logics.

\subsubsection{Sequent Calculi}
A \emph{multi-conclusion sequent} over $\mathcal L$ is $\Gamma\Rightarrow\Delta$, where $\Gamma,\Delta$ are finite multisets of $\mathcal L$-formulas; it is \emph{single-conclusion} if $|\Delta|\le1$. For $S=(\Gamma\Rightarrow\Delta)$, its \emph{interpretation} is $I(S)=\bigast\Gamma\to\bigplus\Delta$, its \emph{antecedent} is $\Gamma$, and its \emph{succedent} is $\Delta$. We write $\Gamma\Leftrightarrow\Delta$ for $\Gamma\Rightarrow\Delta$ and $\Delta\Rightarrow\Gamma$. A \emph{multi-conclusion rule} over $\mathcal L$ has the form $S_1,\dots,S_n / S$ with \emph{premises} $S_1,\dots,S_n$ and \emph{conclusion} $S$; an \emph{axiom} has no premises. A rule is \emph{single-conclusion} if all its sequents are.

A \emph{multi-conclusion sequent calculus} $G$ over $\mathcal L$ is a set of such rules, with language $\mathcal L_G$; it is single-conclusion if all its rules are. For a set of sequents $\calA$, by $G+\calA$ we mean the calculus obtained by adding every sequent in $\calA$ as an axiom to $G$. 
By a  \emph{dag-like proof} $\pi$ in $G$, a \emph{$G$-proof} in short, of a (single-conclusion) sequent $S$ from a set $\mathcal{S}$ of (single-conclusion) sequents, we mean a finite sequence $\pi := S_1, \cdots, S_m$ of sequents such that $S_m = S$ and each $S_i$ is obtained from earlier sequents by an instance of a rule of $G$. A $G$-proof is called \emph{tree-like} if each $S_j$ is used at most once in the proof as a hypothesis of a rule. When $\pi$ is a $G$-proof of $S$ from the set of sequents  $\mathcal{S}$, we write $\mathcal{S} \vdash^{\pi}_G S$. By $\mathcal{S} \vdash_G S$, we mean that there exists a $G$-proof $\pi$ such that $\mathcal{S} \vdash^{\pi}_G S$ and we write $G \vdash S$ for $\varnothing \vdash_G S$. Formulas $A$ and $B$ are \emph{provably equivalent} if $A \Leftrightarrow B$ is provable in $G$. The  \emph{size} of a formula $A$ or a proof $\pi$, denoted $|A|$ and $|\pi|$, is the number of symbols it contains. The \emph{number of lines} of a proof $\pi$, denoted $l(\pi)$, is the total number of formulas appearing in $\pi$. 

Let $L$ be a logic and $G$ a sequent calculus over $\mathcal L\in \{\mathcal L_u,\mathcal L_b,\mathcal L_!\}$. Then $G$ is a \emph{sequent calculus for} $L$ (and $L$ the \emph{logic of} $G$) when
$
G \vdash \Gamma \Rightarrow \Delta \ \text{iff} \ (\bigast \Gamma \to  \mathsmaller{\bigplus} \Delta) \in L$, for every sequent $\Gamma\Rightarrow\Delta$.
In particular, $G\vdash\phi\Rightarrow\psi$ iff $(\phi\to\psi)\in L$. 

We recall standard sequent calculi for many substructural and linear logics in Tables \ref{Fig: sequent calculus} and \ref{Fig: linear Seq cal}.
The single-conclusion versions of the calculi in Table \ref{Fig: linear Seq cal}
are $\BFL$, $\BFLw$, $\BFLc$, $\LJu$, $\BIMALL$, $\BAIMALL$, $\BRIMALL$, $\LJb$, $\BILL$, $\BAILL$, $\BRILL$, and $\LJajab$, respectively. Note that, as we work with multisets, the exchange rules are built in. For any sequent calculus $G$ over $\mathcal{L}$ defined in this paragraph, by $G^{-}$ we mean $G$ without the cut rule. We define $L_G$ as the set $\{A \in \mathcal{L} \mid G \vdash \, \Rightarrow A\}$. It is known that $L_G$ is the logic of $G$ in the sense defined above. We use boldface letters for calculi (e.g.\ $\mathbf{CFL_e}$) and sans-serif for their logics (e.g.\ $L_{\mathbf{CFL_e}}=\mathsf{CFL_e}$). We define the consequence relation $\vdash_{L_G}$ for $L_G$ by setting $\Gamma \vdash_{L_G} A$ as $\{\Rightarrow \gamma \mid \gamma \in \Gamma\} \vdash_G \, \Rightarrow A$, for any subset $\Gamma \cup \{A\}$ of $\mathcal{L}$-formulas. Finally, for any $G$ in Table~\ref{Fig: linear Seq cal} with the logic $L=L_G$, by $iL$ we mean the logic of the single-conclusion version of $G$, denoted by $iG$.
\begin{table}[t]
    \centering
    $ A \Rightarrow A \ \ \scriptsize{\text{(id)}}
        \qquad \qquad
        \Rightarrow 1
        \qquad \qquad
        0 \Rightarrow$ 
           \\[.5ex]
$\vlinf{}{{\mathsmaller{(1w)}}}{\Gamma, 1 \Rightarrow \Delta}{\Gamma \Rightarrow \Delta}
        \qquad
\vlinf{}{{\mathsmaller{(0w)}}}{\Gamma \Rightarrow 0, \Delta}{\Gamma \Rightarrow \Delta}
    $
    \\[.5ex]
    $
\vlinf{(i=0,1)}{\mathsmaller{(L\wedge_i)}}{\Gamma, A_1 \wedge A_2 \Rightarrow \Delta}{\Gamma, A_i\Rightarrow\Delta}
\qquad
\vliinf{}{\mathsmaller{(R\wedge)}}{\Gamma\Rightarrow A\wedge B, \Delta}{\Gamma\Rightarrow A, \Delta}{\Gamma\Rightarrow B, \Delta}
    $ 
 \\[.5ex]
    $
        \vliinf{}{\mathsmaller{(L\vee)}}{\Gamma, A \vee B \Rightarrow\Delta}{\Gamma, A\Rightarrow\Delta}{\Gamma, B\Rightarrow\Delta} 
        \qquad
        \vlinf{(i=0,1)}{\mathsmaller{(R\vee_i)}}{\Gamma\Rightarrow A_0 \vee A_1, \Delta}{\Gamma\Rightarrow A_i, \Delta} 
    $
    \\[.5ex]
    $
    \vlinf{}{\mathsmaller{(L *)}}{\Gamma, A * B \Rightarrow  \Delta}{\Gamma, A, B \Rightarrow \Delta}
        \qquad
        \vliinf{}{\mathsmaller{(R*)}}{\Gamma, \Sigma \Rightarrow A * B, \Delta, \Lambda} {\Gamma \Rightarrow A, \Delta}{\Sigma \Rightarrow B, \Lambda}
    $
    \\[.5ex]
     $
    \vliinf{}{\mathsmaller{(L \to)}}{\Gamma, \Sigma, A \to B \Rightarrow \Delta, \Lambda} {\Gamma \Rightarrow A, \Delta}{\Sigma, B \Rightarrow \Lambda}
    \qquad
    \vlinf{}{\mathsmaller{(R \to)}}{\Gamma \Rightarrow A \to B, \Delta}{\Gamma, A \Rightarrow B, \Delta}
    $
    \\[.5ex]
    $
        \vliinf{}{\mathsmaller{\text{(cut)}}}{\Gamma,\Sigma \Rightarrow \Delta,\Lambda}{\Gamma \Rightarrow A, \Delta}{\Sigma, A \Rightarrow\Lambda}
    $
    \vspace{5pt}
\caption{Sequent calculus $\mathbf{CFL_e}$}
\label{Fig: sequent calculus}
\end{table}

\begin{table}[t]
\centering
    \begin{tabular}{c c}
    \vspace{5pt}
$\Gamma \Rightarrow \top, \Delta \ \ \mathsmaller{(\top)}$ & $\Gamma, \bot \Rightarrow \Delta \ \ \mathsmaller{(\bot)}$ \\ \vspace{5pt}
\AxiomC{$\Gamma \Rightarrow \Delta$}
\RightLabel{\scriptsize{$(Lw)$}}
\UnaryInfC{$\Gamma, A \Rightarrow \Delta$}
\DisplayProof &
\AxiomC{$\Gamma \Rightarrow \Delta$}
\RightLabel{\scriptsize{$(Rw)$}}
\UnaryInfC{$\Gamma \Rightarrow A, \Delta$} 
\DisplayProof
\\
\vspace{5pt}
\AxiomC{$\Gamma, A, A \Rightarrow \Delta$}
\RightLabel{\scriptsize{$(Lc)$}}
\UnaryInfC{$\Gamma, A \Rightarrow \Delta$}
\DisplayProof &
\AxiomC{$\Gamma \Rightarrow A, A, \Delta$}
\RightLabel{\scriptsize{$(Rc)$}}
\UnaryInfC{$\Gamma \Rightarrow A, \Delta$} 
\DisplayProof\\
\vspace{5pt}
\AxiomC{$! \Gamma \Rightarrow A$}
\RightLabel{\scriptsize{$(R!)$}}
\UnaryInfC{$! \Gamma \Rightarrow !A$}
\DisplayProof
&
\AxiomC{$\Gamma , A \Rightarrow \Delta$}
\RightLabel{\scriptsize{$(L!)$}}
\UnaryInfC{$\Gamma , !A \Rightarrow \Delta$}
\DisplayProof
\\
\AxiomC{$\Gamma \Rightarrow \Delta$}
\RightLabel{\scriptsize{$(W!)$}}
\UnaryInfC{$\Gamma , !A \Rightarrow \Delta$}
\DisplayProof
&
\AxiomC{$\Gamma, !A, !A \Rightarrow \Delta$}
\RightLabel{\scriptsize{$(C!)$}}
\UnaryInfC{$\Gamma , !A \Rightarrow \Delta$}
\DisplayProof
 \end{tabular}
 \vspace{5pt}
    \caption{Additional axioms and rules.}
    \label{fig: structural rules}
\end{table}

\begin{table}[t]
\centering
\renewcommand{\arraystretch}{1.1}
\begin{tabular}{|l|}
\hline
Calculi over $\calL_u$:  \\
$\BCFL$ := Table \ref{Fig: sequent calculus} \\
$\BCFLw$ := $\BCFL + \{(Lw), (Rw)\}$  \\
$\BCFLc$ := $\BCFL + \{(Lc), (Rc)\}$  \\
$\LKu:= \BCFLcw$ := $\BCFLc + \{(Lw), (Rw)\}$  \\
\hline
Calculi over $\calL_b$:  \\
$\BMALL$ := $\BCFL + \{(\bot), (\top)\}$  \\
$\BAMALL$ := $\BMALL + \{(Lw), (Rw)\}$  \\
$\BRMALL$ := $\BMALL + \{(Lc), (Rc)\}$  \\ 
$\LKb$ := $\BRMALL + \{(Lw), (Rw)\}$  \\
\hline
Calculi over $\calL_!$:  \\
$\BCLL$ := $\BMALL + \{(R!), (L!), (W!), (C!)\}$  \\
$\BALL$ := $\BCLL + \{(Lw), (Rw)\}$  \\
$\BRLL$ := $\BCLL + \{(Lc), (Rc)\}$  \\
$\LKajab$ := $\BRLL + \{(Lw), (Rw)\}$  \\
\hline
\end{tabular}
\vspace{5pt}
\caption{Sequent calculi. The additional rules are in Table \ref{fig: structural rules}.}
\label{Fig: linear Seq cal}
\end{table}

Recall the sequent calculus $\LK$ for classical propositional logic, $\CPC$,  over the propositional language $\mathcal{L}_p=\{\top, \bot, \wedge, \vee, \to\}$:
\[
\LK:= \{(id), (\top), (\bot), (Lw), (Rw), (Lc), (Rc), (L \circ), (R \circ), (cut)\},
\]
for $\circ \in \{\wedge, \vee, \to\}$, using the rules in Tables \ref{Fig: sequent calculus} and \ref{fig: structural rules}. The single conclusion version of $\LK$ is denoted by $\LJ$ and is the standard sequent calculus for intuitionistic propositional logic, $\IPC$. 

It is straightforward, yet important, that for any two calculi among $\LK$, $\LKu$, $\LKb$, $\LKajab$, proofs of sequents in the smaller language can be translated between the calculi with only a polynomial increase in proof size and number of lines; likewise for $\LJ$. This follows from the fact that, in the presence of the weakening and contraction rules, $A * B$ is provably equivalent to $A \wedge B$, 0 and 1 correspond to 
$\bot$ and $\top$, and the modalities $!$ and $?$ can be interpreted as the identity, i.e., $!A:=?A:=A$, thereby validating all rules.

\subsubsection{Frege Systems}

We define Frege systems for substructural and linear logics. An \emph{inference system} $F$ over $\mathcal{L}$ consists of rules of the form $A_1, \ldots, A_n / A$, where the $A_i$'s and $A$ are formulas in $\mathcal{L}$. A \emph{dag-like $F$-proof} ($F$-proof, in short) $\pi$ of an $\mathcal{L}$-formula $A$ from a set $\mathcal{A}$ of $\mathcal{L}$-formulas is a sequence $\pi := A_1, \ldots, A_m$ of $\mathcal{L}$-formulas such that $A_m=A$ and each $A_i$ either belongs to $\mathcal{A}$ or is obtained from earlier formulas by an instance of a rule of $F$. An $F$-proof is called \emph{tree-like} if each $A_i$ is used at most once in the proof as a hypothesis of a rule. Each $A_i$ is called a \emph{line} of $\pi$. The number of lines in an $F$-proof $\pi$, denoted by $l(\pi)$, clearly satisfies $l(\pi) \leq |\pi|$. When $\pi$ is an $F$-proof of $A$ from the set of $\mathcal{L}$-formulas $\mathcal{A}$, we write $\mathcal{A} \vdash^{\pi}_F A$. If there exists an $F$-proof $\pi$ such that $\mathcal{A} \vdash^{\pi}_F A$, we write $\mathcal{A} \vdash_F A$. If $\mathcal{A} = \emptyset$, we say that $A$ is \emph{provable} in $F$ and write $F \vdash A$. The \emph{logic of $F$} is the set of formulas provable in $F$.


Let $L$ be a substructural or linear logic introduced in the previous subsection. A finite inference system $F$ is called a \emph{Frege system} for $L$, briefly an $L$-$\F$ system, if it is \emph{sound} (i.e., $\vdash_{F} A$ implies $A \in L$) and \emph{strongly complete} (i.e., $A_1, \ldots, A_n \vdash_{L} A$ implies $A_1, \ldots, A_n \vdash_{F} A$). A Frege system is \emph{standard} if $A_1, \ldots, A_n \vdash_{F} A$ implies $A_1, \ldots, A_n \vdash_{L} A$. Following \cite{jerabek}, we adopt the convention that all Frege systems are standard. As a consequence, any two Frege systems for the same logic $L$ are equivalent, and we refer simply to \emph{the} Frege system for $L$, denoted $L$-$\F$.

The following lemma shows that the choice of the Frege system or sequent calculus does not affect the proof complexity for substructural and linear logics. The proof follows the classical argument of \cite{Cook} and its sequent-calculus variant is essentially \cite[Theorem 6.3]{Raheleh}.

\begin{lemma} \label{Lem: Equivalence Of Frege}
Let $L$ be a substructural or linear logic introduced in this paper. Then, all Frege systems for $L$, as well as Frege systems and sequent calculi for 
$L$, are equivalent.
\end{lemma}

The Frege system for $\mathsf{FL_e}$, denoted by $\FL$-$\F$, is presented in Table \ref{Fig: FL_e} (see \cite[Figure 2.9]{Ono}). Note that the adjunction rule ($A \ B / A \wedge B$) is provable in $\FL$-$\F$. Adding the additional axioms and rules in Table \ref{Fig: linear}, we obtain Frege systems for several substructural and linear logics. 
Table \ref{Fig: contraction-free logics} presents Frege systems for classical contraction-free logics.
\begin{table}[t!]
\centering
\begin{tabular}{|c|c|}
\hline (id) & $A \rightarrow A$ \\
\hline (pf) & $(A \rightarrow B) \rightarrow \big((C \rightarrow A) \rightarrow(C \rightarrow B) \big)$  \\
\hline (per) & $\big(A \rightarrow(B \rightarrow C)\big) \rightarrow \big(B \rightarrow(A \rightarrow C) \big)$  \\
\hline $(* \wedge)$ & $\big((A \wedge 1) * (B \wedge 1)\big) \rightarrow(A \wedge B)$ \\
\hline $(\wedge \! \to)_1$ & $(A \wedge B) \rightarrow A$ \\
\hline $(\wedge \! \rightarrow)_2$ & $(A \wedge B) \rightarrow B$  \\
\hline $(\rightarrow \! \wedge)$ & $\big((A \rightarrow B) \wedge(A \rightarrow C)\big) \rightarrow \big(A \rightarrow(B \wedge C) \big)$ \\
\hline $(\rightarrow \! \vee)_1$ & $A \rightarrow(A \vee B)$  \\
\hline $(\rightarrow \! \vee)_2$ & $B \rightarrow(A \vee B)$  \\
\hline $(\vee \! \rightarrow)$ & $\big((A \rightarrow C) \wedge(B \rightarrow C)\big) \rightarrow (A \vee B \rightarrow C)$ \\
\hline $(\rightarrow \! *)$ & $B \rightarrow(A \rightarrow (A \st B))$  \\
\hline $(* \! \rightarrow)$ & $\big(B \rightarrow(A \rightarrow C)\big) \rightarrow((A \st B) \rightarrow C)$  \\
\hline (1) & 1  \\
\hline $(1 \! \rightarrow)$ & $1 \rightarrow(A \rightarrow A)$  \\
\hline $(\mathrm{mp})$ & $\vliinf{}{}{B}{A}{A \to B}$
 \\
\hline $(\operatorname{adj}_u)$ &   $\vlinf{}{}{A \wedge 1}{A}$  \\
\hline
\end{tabular}
\vspace{5pt}
\caption{The system $\FL$-$\F$.}\label{Fig: FL_e}
\end{table}
\begin{table}[t!]
\centering
\begin{tabular}{ |c| c| c| c| }
\hline
$(\text{top})$ & $A \to \top$\\
\hline
$(\text{bot})$ & $ \bot \to A$ \\
\hline
$(\text{dn})$ & $\neg \neg A \to A$ \\
\hline
$(w)$ & $A \to (B \to A)$ \\
\hline
$(c)$ & $(A \to (A \to B)) \to (A \to B)$ \\
\hline
$(!w)$ & $A \to (!B \to A)$\\
\hline
$(!c)$ & $(!A \to (!A \to B)) \to (!A \to B)$\\
\hline
$(!\text{K})$ & $!(A \to B) \to (!A \to !B)$\\
\hline
$(!\text{T})$ & $!A \to A$\\
\hline
$(!\text{4})$ & $!A \to !!A$\\
\hline
$(\text{nec})$ & 
$\vlinf{}{}{!A}{A}$
\\
\hline
\end{tabular}
\vspace{5pt}
\caption{{Additional axioms and rules.}}\label{Fig: linear}
\end{table}
\begin{table}[t]
\centering
\renewcommand{\arraystretch}{1.1}
\begin{tabular}{|l|}
\hline
Frege systems over $\calL_u$:  \\
$\CFL$-$\F$ := $\FL$-$\F$ + \{(dn)\} \\
$\CFLw$-$\F$ := $\CFL$-$\F$ + $\{(w)\}$  \\
\hline
Frege systems over  $\calL_b$:  \\
$\MALL$-$\F$ := $\CFL$-$\F$ + \{(top), (bot)\}  \\
$\AMALL$-$\F$ := $\MALL$-$\F$ + $\{(w)\}$  \\ 
\hline
Frege systems over  $\calL_!$:  \\
$\CLL$-$\F$ := $\MALL$-$\F$ + $\{(!w), (!c), (!\text{K}), (!\text{T}), (!4), \text{(nec)}\}$  \\
$\ALL$-$\F$ := $\CLL$-$\F$ + $\{(w)\}$  \\
\hline
\end{tabular}
\vspace{5pt}
\caption{Frege systems for classical contraction-free logics.}
\label{Fig: contraction-free logics}
\end{table}

\section{Clique-Color Formulas}\label{sec: HardTheorems}


In this section, we introduce some known hard theorems for classical and non-classical proof systems that we will use later. 



Let $n, k, m \geq 1$. A simple, undirected graph without loops (simply, a graph) contains a $k$-\emph{clique} if it has a complete subgraph on $k$ vertices, and it is \emph{$m$-colorable} if there exists an assignment of $m$ colors to its vertices such that adjacent vertices receive distinct colors.
These properties of graphs can be represented in propositional language. To illustrate this, first consider the $\binom{n}{2}$ atoms $p_{i,j}$, where $i \neq j \in [n]=\{1, \ldots, n\}$. Each Boolean assignment to these atoms represents a graph on the vertex set $[n]$. The intended meaning is that the atom $p_{i,j}$ has value $1$ if and only if the vertices $i$ and $j$ are connected by an edge in the graph.

Then, to represent a $k$-clique, we use $kn$ additional atoms $r_{u,i}$, where $u \in [k]$ and $i \in [n]$. Any assignment to these atoms represents a function from $[k]$ to $[n]$. More precisely, the atom $r_{u,i}$ has value $1$ if and only if the function maps $u$ to $i$.
Now consider the following clauses:
\begin{itemize}
\item[$\bullet$] 
$\bigvee_{i \in [n]} r_{u,i}$, for all $u \leq k$,
\item[$\bullet$] 
$\neg r_{u,i} \vee \neg r_{u,j}$, for all $u \in [k]$ and any $i \neq j \in [n]$,
\item[$\bullet$] 
$\neg r_{u,i} \vee \neg r_{v,i}$, for all $u \neq v \in [k]$ and any $i \in [n]$,
\item[$\bullet$] 
$\neg r_{u,i} \vee \neg r_{v,j} \vee p_{i,j}$, for all $u \neq v \in [k]$ and $i \neq j \in [n]$.
\end{itemize}
Let $Clique^k_n(\bar{p}, \bar{r})$ denote the conjunction of all these formulas. The intended meaning of $Clique^k_n(\bar{p}, \bar{r})$ is that the atoms $\bar{p}$ represent a graph on $n$ vertices, while $\bar{r}$ encodes an injective function from $[k]$ to $[n]$, selecting the vertices of a $k$-clique in the graph described by $\bar{p}$. Observe that every occurrence of $p_{i,j}$ in $Clique^k_n(\bar{p}, \bar{r})$ is positive; equivalently, the formula is \emph{monotone} in~$\bar{p}$.

For $m$-colorability, we use $nm$ additional atoms $s_{i,a}$ for $i \in [n]$ and $a \in [m]$, indicating that vertex $i$ receives color $a$. Using these atoms and atoms $\bar{p}$, consider the clauses:
\begin{itemize}
\item[$\bullet$] $\bigvee_{a \in [m]} s_{i,a}$, for all $i \in [n]$,
\item[$\bullet$] $\neg s_{i,a} \vee \neg s_{i,b}$, for all $a \neq b \in [m]$ and $i \in [n]$,
\item[$\bullet$] $\neg s_{i,a} \vee \neg s_{j,a} \vee \neg p_{i,j}$, for all $a \in [m]$ and $i \neq j \in [n]$.
\end{itemize}
Let $Color^m_n(\bar{p}, \bar{s})$ denote the conjunction of these formulas. The intended meaning of $Color^m_n(\bar{p}, \bar{s})$ is that the atoms $\bar{s}$ encode a function from $[n]$ to $[m]$, assigning colors to the vertices in such a way that adjacent vertices receive distinct colors. This completes the representability.

Now, as no graph with $k+1$-clique can have a $k$-coloring, it is clear that the formula 
\[
Clique^{k+1}_n(\bar{p}, \bar{r}) \to \neg Color^{k}_n(\bar{p}, \bar{s})  
\]
is a classical tautology. We call this formula the \emph{Clique-Color formula}.
Buss~\cite{BussPoly} provided a polynomial-size tree-like $\LK$-proof of the pigeonhole principle. Since the proof of the Clique--Color formula uses the pigeonhole principle in a simple manner, it is known that the Clique--Color formula has a polynomial-size tree-like $\LK$-proof.

\begin{theorem}\cite{BussPoly,Krajicek}\label{Thm: Buss Clique Color}
There is a tree-like $\LK$-proof $\pi_n$ for
\[
Clique^{k+1}_n(\bar{p}, \bar{r}) \to \neg Color^{k}_n(\bar{p}, \bar{s})  
\]
such that $|\pi_n| \leq n^{O(1)}$. 
\end{theorem}

The significance of the Clique-Color formula stems from the following hardness theorem in monotone circuit complexity:

\begin{theorem}\cite{Alon,Razborov} \label{Thm: Alon}
For $k = \lfloor \sqrt{n} \rfloor$, the Clique--Color formula requires 
monotone interpolating circuits\footnote{As circuits are not directly used in this paper, we do not define them here; see~\cite{Krajicek}.} of size $2^{\Omega(n^{1/4})}$. More precisely, assume that $C$ is a circuit consisting only of $\{\wedge, \vee\}$ gates such that for any $\bar u \in \{0, 1\}^{\binom{n}{2}}$:
\begin{itemize}
    \item if $C(\bar{u}) = 0$, then $Clique^{k+1}_n(\bar{u}, \bar{r})$ is unsatisfiable, and
    \item if $C(\bar{u}) = 1$, then $Color^k_n(\bar{u}, \bar{s})$ is unsatisfiable.
\end{itemize}
Then, $C$ must contain at least $2^{\Omega(n^{1/4})}$ gates.
\end{theorem}

This hardness theorem has been repeatedly used to establish proof size lower bounds in various proof systems. The only result of this type that is relevant for the present paper is the exponential lower bound for $\LK^-$. To state this result, we first need to introduce (monotone) feasible interpolation.

\begin{definition}[Kraj\'{i}\v{c}ek \cite{krajivcekfeasible2}]
A proof system $P$ for $\mathsf{CPC}$ is said to have \emph{(monotone) feasible interpolation} if for any $P$-proof $\pi$ of an implication $A(\bar{p}, \bar{q}) \to B(\bar{p}, \bar{r})$, (resp. where $A$ or $B$ is monotone in $\bar{p}$), there is a (monotone) circuit $C$ of size $(|\pi|+|A|+|B|)^{O(1)}$ computing a Craig interpolant for $A(\bar{p}, \bar{q}) \to B(\bar{p}, \bar{r})$.
\end{definition}

\begin{theorem}\cite[Theorems 3.3.1, 3.3.2]{Krajicek}\label{thm: FI For LKminus}
$\LKminus$ has the monotone feasible interpolation property.  
\end{theorem} 

Combining Theorems~\ref{Thm: Alon} and~\ref{thm: FI For LKminus}, it is straightforward to obtain an exponential lower bound for $\LK^-$-proofs of the Clique--Color formula. We will adapt this strategy in Section~\ref{Sec: cut-free}.

For non-classical logics, by modifying the Clique--Color formulas so that they become intuitionistic tautologies rather than classical tautologies and providing a corresponding notion of feasible (monotone) interpolation, Hrube\v{s} \cite{Hrubes} established an exponential lower bound on the \emph{number of lines} for $\LJ$-proofs and intuitionistic Frege proofs:

\begin{theorem}\label{Hrubes} \cite{Hrubes}
Let $\bar{p}=\{p_{i,j} \mid 1 \leq i \neq j \leq n\}$, $\bar{q}=\{q_{i,j} \mid 1 \leq i \neq j \leq n\}$, $\bar{r}$, and $ \bar{s}$ be pairwise disjoint sets of variables. Let $k = \lfloor \sqrt{n} \rfloor$. Then the formulas
\[
\Theta^{\bot}_n := \bigwedge_{1 \leq i \neq j \leq n} (p_{i,j} \vee q_{i,j}) \to \neg Color^{k}_n(\bar{p}, \bar{s}) \vee \neg Clique^{k+1}_n(\neg \bar{q}, \bar{r})
\]
are intuitionistic tautologies and every $\LJ$-proof of $\Theta^{\bot}_n$ contains at least $2^{\Omega(n^{1/4})}$ lines.
\end{theorem}

The superscript $\bot$ in $\Theta^{\bot}_n$ highlights the presence of negations in the formula. To move further into the substructural setting, we require a negation-free variant of $\Theta^{\bot}_n$, introduced by Jeřábek~\cite{jerabek}. The main idea is to introduce fresh variables $s'_{i,l}$ and $r'_{i,l}$ as placeholders for $\neg s_{i,l}$ and $\neg r_{i,l}$, respectively:

\begin{definition} \cite[Definition 6.28]{jerabek} \label{def: alpha beta}
For $k \leq n$, let
\[
\alpha^k_n (\bar{p}, \bar{s}, \bar{s'}) := \bigvee_{i\in [n]} \bigwedge_{l \in [k]} s'_{i,l} \;\;\vee\;\; \bigvee_{i\neq j \in [n]} \bigvee_{l \in [k]} (s_{i,l} \wedge s_{j,l} \wedge p_{i,j}),
\]
\[
\beta^k_n (\bar{q}, \bar{r}, \bar{r'}) := \bigvee_{l \in [k]} \bigwedge_{i \in [n]} r'_{i,l} \;\;\vee\;\; \bigvee_{i \neq j \in [n]} \bigvee_{l \neq m \in [k]} (r_{i,l} \wedge r_{j,m} \wedge q_{i,j}).
\]
Define the negation-free version of Hrube\v{s}'s formulas by
\[
\Theta_{n,k}:=  \big(\bigwedge_{i,j} (p_{i,j} \vee q_{i,j})\big) \; \to 
\]
\[
[
(\bigwedge_{i,l}(s_{i,l} \vee s'_{i,l}) \to \alpha^k_n (\bar{p}, \bar{s}, \bar{s'})) \vee  (\bigwedge_{i,l}(r_{i,l} \vee r'_{i,l}) \to \beta^{k+1}_n (\bar{q}, \bar{r}, \bar{r'}))].
\]
Note that
\[
 Color^k_n(\bar{p}, \bar{s}) = \neg \alpha^k_n(\bar{p}, \bar{s}, \neg \bar{s}) \qquad 
Clique^k_n(\bar{p}, \bar{r}) = \neg \beta^k_n(\neg \bar{p}, \bar{r}, \neg \bar{r}).
\] 
\end{definition}

\begin{theorem}(\cite[Theorem 6.37]{jerabek}) \label{JerabekAsli}
The formulas $\Theta_n := \Theta_{n,\lfloor \sqrt{n} \rfloor}$ are intuitionistic tautologies and require $\LJ$-proofs with at least $2^{n^{\Omega(1)}}$ lines.
\end{theorem}

Since we work in a substructural/linear setting, we require formulas provable not only in intuitionistic logic but also in weaker logics such as $\FL$. For this purpose, we use the formulas $\Theta^*_{n,k}$ introduced in \cite[Theorem 4.11]{Raheleh}\footnote{In \cite{Raheleh}, the formulas $\Theta^{*}_{n,k}$ are defined using $\setminus$ instead of $\to$, since the logics considered there do not necessarily admit exchange.}:
\[
\Theta^{*}_{n,k}:= [\bigast_{i,j} ((p_{i,j} \wedge 1) \vee (q_{i,j} \wedge 1))] \; \to 
\]
\[
\big([\bigast_{i,l} ((s_{i,l} \wedge 1) \vee (s'_{i,l} \wedge 1)) \to \alpha^k_n (\bar{p}, \bar{s}, \bar{s'})] 
\vee [\bigast_{i,l}((r_{i,l} \wedge 1) \vee  (r'_{i,l} \wedge 1)) \to \beta^{k+1}_n (\bar{q}, \bar{r}, \bar{r'})]\big).
\]

\begin{theorem}\cite[Corollary 5.4.]{Raheleh} \label{Thm: exp lower bound FLe}
The formulas $\Theta^{*}_{n,k}$ are provable in $\BFL$.
\end{theorem}

The following corollary collects all the results presented in this section in a unified substructural setting that we will use later.

\begin{corollary}\label{RemarkLJAjab}
The formula $\Theta^{*}_{n}:=\Theta^{*}_{n,\lfloor \sqrt{n} \rfloor}$ is $\BFL$-provable and admits an $\LKu$-proof of size $n^{O(1)}$, while any $\LJajab$-proof of it has at least $2^{n^{\Omega(1)}}$ proof lines.
\end{corollary}
\begin{proof}
The $\BFL$-provability is a direct consequence of Theorem \ref{Thm: exp lower bound FLe}. For the upper bound, since $\Theta_{n, k}$ is a simple variant of the Clique--Color formula and $k \leq n$, the sequent
\[
Clique^{k+1}_n(\bar{p}, \bar{r}) \to \neg Color^{k}_n(\bar{p}, \bar{s}) \Rightarrow \Theta_{n,k}
\]
has an $\LK$-proof of size $n^{O(1)}$.
Therefore, by setting $k=\lfloor \sqrt{n} \rfloor$ and Theorem~\ref{Thm: Buss Clique Color}, the sequent $\Rightarrow \Theta_n$ has an $\LK$-proof and hence an $\LKu$-proof of size $n^{O(1)}$. Finally, using the structural rules in $\LKu$, it is straightforward to obtain an $\LKu$-proof of $\Theta_n \Rightarrow \Theta_n^*$ of size $n^{O(1)}$. Consequently, we obtain an $\LKu$-proof of $\Rightarrow \Theta_n^*$ of size $n^{O(1)}$.

For the lower bound, let $\pi_n$ be an $\LJajab$-proof of $\Theta^*_n$. Using the weakening and contraction rules, one can easily construct an $\LJajab$-proof of $\Theta_n^* \Rightarrow \Theta_n$ with an $n^{O(1)}$ many number of lines. Hence, one obtains an $\LJajab$-proof of $\Rightarrow \Theta_n$ with $l(\pi_n)+n^{O(1)}$ lines. By erasing $!$-occurrences in this proof, one obtains an $\LJ$-proof of $\Rightarrow \Theta_n$ with $l(\pi_n)+n^{O(1)}$  lines. Therefore, by Theorem~\ref{JerabekAsli}, we conclude that $l(\pi_n) \geq 2^{n^{\Omega(1)}}$.
\end{proof}

\section{Affine Linear Logic}\label{Sec: Affine Linear Logic}

In this section, we prove the first of our main theorems. Namely, we identify a sequence of
$\mathbf{FL_e}$-provable formulas admitting polynomial-size $\LKu$-proofs whose shortest $\BALL$-proofs are exponentially long. 
Our strategy is to use a proof transformation that we call \emph{Chu's translation} to establish a form of feasible conservativity of $\BALL$ over $\BAILL$ for a class of formulas that includes $\Theta^*_n$, as introduced in Section~\ref{sec: HardTheorems}. Then, since any $\BAILL$-proof of a $!$-free formula is also an $\LJajab$-proof, we can apply Corollary~\ref{RemarkLJAjab}.

\subsection{Chu's Translation}
Chu's construction is a categorical recipe for building a $*$-autonomous category (i.e., a categorical model for the multiplicative fragment of $\mathbf{MALL}$) 
from a given symmetric monoidal closed category $\mathcal{C}$ (i.e., a categorical model for the multiplicative fragment of $\mathbf{IMALL}$) with a chosen ``dualizing'' object. Intuitively, it enriches $\mathcal{C}$ with a built-in duality, which is exactly the feature needed to model 
involution (see \cite{Barr,Chu}, \cite[Section 3.3]{Hyland}).

The following translation is inspired by Chu's construction, extended to cover full linear logic \cite[Theorem~24]{Hyland}. In fact, it can be seen as the result of applying this construction to the Lindenbaum algebra of an intuitionistic linear logic in order to obtain a model for its classical counterpart. 

\begin{definition}[Chu's translation]\label{Def: translation}
Fix formulas $D,N \in \mathcal{L}_!$ and define the translation functions $(\cdot)^t:\mathcal{L}_! \to \mathcal{L}_!$ and $(\cdot)^s: \mathcal{L}_! \to \mathcal{L}_!$ simultaneously and inductively as follows:\medskip

\noindent $\circ$ Translation $(\cdot)^t$:
\begin{itemize}
\item 
$p^t:=p$ \quad where $p$ is an atom
\item 
$1^t:=1$ \qquad $0^t:=D$ \qquad $\top^t:=\top$ \qquad $\bot^t:=\bot$
\item 
$(A \circ B)^t:=A^t \circ B^t$  for $\circ \in \{\wedge, \vee, *\}$
\item 
$(A \to B)^t:=(A^t \to B^t) \wedge (B^s \to A^s)$
\item 
$(!A)^t:=!A^t$ 
\end{itemize}

\noindent $\circ$ Translation $(\cdot)^s$:
\begin{itemize}
\item 
$p^s:= N$ \quad where $p$ is an atom 
\item 
$1^s:= D$ \qquad $0^s:= 1$ \qquad $\top^s:= \bot$ \qquad $\bot^s:= \top$ 
\item 
$(A \wedge B)^s:= A^s \vee B^s$ \qquad $(A \vee B)^s:= A^s \wedge B^s$
\item
$(A * B)^s:= (A^t \to B^s) \wedge (B^t \to A^s)$ 
\item 
$(A \to B)^s:= A^t \st B^s$ \qquad $(!A)^s:= !A^t \to D$
\end{itemize}
For any $\mathcal{L} \in \{\mathcal{L}_u, \mathcal{L}_b\}$, if $N, D \in \mathcal{L}$, the translations can be defined analogously over $\mathcal{L}$ by replacing $\mathcal{L}_!$ with $\mathcal{L}$ and restricting to formulas in $\mathcal{L}$.
\end{definition}
Note that the translations $(\cdot)^t$ and $(\cdot)^s$ are parameterised by $D$ and $N$, but we do not indicate this explicitly in order to keep the notation light. The intended choices of $D$ and $N$ will always be clear from the context.

\begin{remark}
The translations $(\cdot)^t$ and $(\cdot)^s$ may grow exponentially. Indeed, the clauses for $(A \to B)^t$ (resp.\ $(A * B)^s$) involve $A^t$, $A^s$, $B^t$, and $B^s$, which in effect cause the size to roughly double at each translation of an implication (resp.\ multiplication).   
\end{remark}



Let $\calL \in \{\calL_u, \calL_b, \calL_!\}$ and 
\[
L \in \{\CFL, \MALL, \AMALL, \CLL, \ALL\}
\]
with sequent calculus $G$. We write 
\[
iG_D:=iG+\{p, N \Rightarrow D \mid p \; \text{an atom}\},
\]
where $iG$ is the single-conclusion version of $G$.
The following lemma highlights a key property of Chu’s translation, namely that the translations $(\cdot)^t$ and $(\cdot)^s$ play \emph{dual} roles.

\begin{lemma}\label{Lem: iGD}
Let $\calL \in \{\calL_u, \calL_b, \calL_!\}$, let $A$ be an $\calL$-formula and $L \in \{\CFL, \allowbreak \MALL, \CLL\}$ with the calculus $G$. The following sequents have $iG_D$-proofs with $O(|A|)$ many lines:
\[
A^t, A^s \Rightarrow D \qquad (\neg A)^t \Leftrightarrow A^s \qquad (\neg A)^s \Leftrightarrow A^t.
\]
\end{lemma}
\begin{proof} 

We treat the case $\mathcal{L}=\mathcal{L}_!$ and $L=\mathsf{LL}$. The remaining cases follow by the same argument, simply by restricting the language. We begin by showing that $iG_D \vdash A^t, A^s \Rightarrow D$. 
The proof proceeds by induction on the structure of the formula $A$. We first construct the proof in each case and then explain the upper bound for the number of lines at the end. If $A$ is an atom or a constant then $A^t, A^s \Rightarrow D$ will be of one of the forms $p, N \Rightarrow D$, or $\top, \bot \Rightarrow D$, or $1, D \Rightarrow D$, which are all provable in $iG_D$. 
For $A=B \circ C$, where $\circ \in \{\wedge, \vee, *, \to\}$, 
by the induction hypothesis for the formulas $B$ and $C$ we have:
\[
iG_D \vdash B^t, B^s \Rightarrow D \qquad iG_D \vdash C^t, C^s \Rightarrow D
\]
We only address the cases $\circ \in \{\wedge, \to\}$; the remaining cases are similar.
Let $A=B \wedge C$. By Definition \ref{Def: translation}, $A^t=B^t \wedge C^t$ and $A^s=B^s \vee C^s$.
Therefore, it is enough to consider
the following proof in $iG_D$ for $A^t, A^s \Rightarrow D$:
\begin{prooftree}
\AxiomC{$B^t, B^s \Rightarrow D$}
\UnaryInfC{$B^t \wedge C^t, B^s \Rightarrow D$}
\AxiomC{$C^t, C^s \Rightarrow D$}
\UnaryInfC{$B^t \wedge C^t, C^s \Rightarrow D$}
\BinaryInfC{$B^t \wedge C^t, B^s \vee C^s \Rightarrow D$}
\end{prooftree}
For $A=B \to C$, as $A^t=(B^t \to C^t) \wedge (C^s \to B^s)$ and $A^s=B^t * C^s$, we have the following proof: 
\begin{prooftree}
\AxiomC{$C^s \Rightarrow C^s$}
\AxiomC{$B^t, B^s \Rightarrow D$}
\BinaryInfC{$C^s \to B^s, B^t, C^s \Rightarrow D$}
\UnaryInfC{$(B^t \to C^t) \wedge (C^s \to B^s), B^t, C^s \Rightarrow D$}
\UnaryInfC{$(B^t \to C^t) \wedge (C^s \to B^s), B^t* C^s \Rightarrow D$}
\end{prooftree}
In the case that $A=!B$, as $(!B)^t=!B^t$, $(!B)^s=!B^t \to D$, and $iG_D \vdash !B^t, !B^t \to D \Rightarrow D$, we clearly have $iG_D \vdash (!B)^t, (!B)^s \Rightarrow D$.

Finally, regarding the number of lines, our construction yields proofs with a constant number of lines in the axiom cases. In each inductive step leading to a formula $A$, in the worst case we extend the proofs obtained from the induction hypothesis by only a constant number of formulas. Since the construction of the proof requires $O(|A|)$ inductive steps, the resulting proof has $O(|A|)$ lines.

To prove the other claims, i.e., $iG_D \vdash (\neg A)^t \Leftrightarrow A^s$ and $iG_D \vdash (\neg A)^s \Leftrightarrow A^t$, as $\neg A= A \to 0$, by Definition \ref{Def: translation} we have
\[
(\neg A)^t= (A^t \to D) \wedge (1 \to A^s) \qquad (\neg A)^s= A^t * 1.
\]
Now, consider the
the following proofs in $iG_D$:
\begin{prooftree}
\AxiomC{$\Rightarrow 1$}
\AxiomC{$A^s \Rightarrow A^s$}
\BinaryInfC{$1 \to A^s \Rightarrow A^s$}
\UnaryInfC{$(A^t \to D) \wedge (1 \to A^s) \Rightarrow A^s$}
\end{prooftree}
and 
\begin{prooftree}
\AxiomC{$A^s, A^t \Rightarrow D$}
\UnaryInfC{$A^s \Rightarrow A^t \to D$}
\AxiomC{$A^s \Rightarrow A^s$}
\UnaryInfC{$A^s, 1 \Rightarrow A^s$}
\UnaryInfC{$A^s \Rightarrow 1 \to A^s$}
\BinaryInfC{$A^s \Rightarrow (A^t \to D) \wedge (1 \to A^s)$}
\end{prooftree}
where the provability of the sequent $A^s, A^t \Rightarrow D$ follows from the previous claim.
Thus, $iG_D \vdash (\neg A)^t \Leftrightarrow A^s$. Also, we have $iG_D \vdash F * 1 \Leftrightarrow F$, for any $\calL$-formula $F$. Hence, $iG_D \vdash (\neg A)^s \Leftrightarrow A^t$. The linear upper bound on the number of lines in each case is straightforward.
\end{proof}

The translation $(\cdot)^t$ yields a method to transfer provable sequents from classical linear and substructural logics to their intuitionistic counterparts. More precisely, for any logic $L \in \{\CFL, \MALL, \CLL\}$ with the sequent calculus $G$, 
if $L \vdash A$ then $iG_D \vdash (\, \Rightarrow A^t)$, for any formula $A$. 
For our proof complexity purposes, we also need to control the complexity of this process. As the translations of formulas may grow exponentially, the natural complexity measure is the \emph{number of lines} in $iG_D$. However, as we will see, this number of lines depends on the \emph{size} of the Frege proof $\pi$ of $A$ instead of the number of lines of $\pi$. It is worth mentioning that a similar argument works if one uses sequent-style $G$-proofs. The reason we chose Frege systems is that their small number of rules makes the translation of proofs easier to control.

\begin{restatable}{theorem}{ThmCFLeToFLe} \label{thm: CFLe to FL_e}
Let $L \in \{\CFL, \MALL, \CLL\}$ with the sequent calculus $G$ and Frege system $L$-$\F$. Then, for any $L$-$\F$ proof $\pi$ of $A$, there exists an $iG_D$-proof $\sigma$ of $(\, \Rightarrow A^t)$ such that $l(\sigma) \leq O(|\pi|)$. In case $D=\bot$, the claim also holds for any $L \in \{\AMALL, \ALL\}$.  
\end{restatable}
\begin{proof}
The proof is a straightforward but tedious induction on $\pi$ and is therefore deferred to Section~\ref{Sec: Detailed Proofs}.
\end{proof}

\begin{remark}
Ignoring the upper bounds, Theorem~\ref{thm: CFLe to FL_e} is a consequence of applying Chu's construction to the Lindenbaum algebra of $iG_D$, taking the equivalence class of $D$ as the dualizing object. To keep the paper purely proof-theoretic, and thus more accessible, and to maintain control over the complexity measures, we chose to present the proof-theoretic proof.
\end{remark}

\subsection{Feasible Conservativity and Separation}\label{Sec: Feasible conservativity} 
In this subsection, we establish a form of feasible conservativity for sequent calculi of the logics in
$
\{\MALL, \AMALL, \CLL, \ALL\}
$
over their intuitionistic counterparts, for suitable classes of formulas as described below. We then use this form of conservativity to prove the promised separation result as the main theorem of this section.

\begin{definition} \label{Def: Conservative formulas}
Define the following sets of $\calL_!$-formulas simultaneously:
\begin{itemize}
\item 
(\emph{Fully conservative}) $B:= p \mid 1 \mid \top \mid B \wedge B \mid B \vee B \mid C \to B$, where $p$ is an atom and $C$ is conservative
\item 
(\emph{Conservative}) $C := \bot \mid B \mid C \wedge C \mid C \vee C \mid C * C \mid !C$, where $B$ is fully conservative.
\end{itemize}
A formula in the languages $\calL_b$ is called fully conservative or conservative if it is so in the extended language $\calL_!$. 
\end{definition}

\begin{remark}
Note that as any fully conservative formula is also conservative, the set of fully conservative formulas is closed under implication. \end{remark}

\begin{example}\label{Exam: conservative}
The formulas $p \wedge q$, $p \vee q$, and $p \to q$, where $p$ and $q$ are atoms, are all fully conservative, while $p * q$ and $!p$ are conservative but not fully conservative. Note that $(p * q) \to r$ and $!p \to r$ are fully conservative. Moreover, $0$ is neither conservative nor fully conservative and $\bot$ is conservative but not fully conservative.
\end{example}

The following lemma shows that if $D = N = \bot$, then the translation $t$ does not change a conservative formula, and if the formula is also fully conservative, then $s$ maps it to falsehood, $\bot$. We hope this clarifies our terminology.

\begin{lemma}\label{lem: basic conservative}
Fix $D=N=\bot$ and let $B$ be a fully conservative and $C$ be a conservative $\calL_b$-formula. Then, there are $\pi_1$, $\pi_2$, $\sigma$, $\tau_1$, and $\tau_2$ such that:
\begin{itemize}
    \item 
$\BIMALL \vdash^{\pi_1} B^t \Rightarrow B$ and $\BIMALL \vdash^{\pi_2} B \Rightarrow B^t$,
\item 
$\BIMALL \vdash^{\sigma} B^s \Rightarrow \bot$,
\item 
$\BIMALL \vdash^{\tau_1} C^t \Rightarrow C$ and $\BIMALL \vdash^{\tau_2} C \Rightarrow C^t$, and
\end{itemize}
$l(\pi_i) \leq O(|B|^2)$, $ l(\sigma) \leq O(|B|)$, and $l(\tau_i) \leq O(|C|^2)$, for any $i \in \{1, 2\}$.
If $B$ and $C$ are $\calL_!$-formulas, then the same results hold when replacing $\BIMALL$ by $\BILL$.
\end{lemma}

\begin{proof}
We only present the case of $\BIMALL$; the case of $\BILL$ is similar. For $\BIMALL$, we construct $\pi_i$'s, $\sigma$, and $\tau_i$'s by a simultaneous recursion on the structure of the formulas $B$ and $C$. We then show that the number of lines of the constructed proofs is bounded.

The construction of proofs in the case where $B$ (resp. $C$) is an atom or a constant is trivial, since by Definition~\ref{Def: translation} and Definition~\ref{Def: Conservative formulas} we have $B^t=B$, $B^s=\bot$, and $C^t=C$. The construction of proofs for $B = B_1 \circ B_2$, where $\circ \in \{\wedge, \vee\}$, is straightforward using Definition~\ref{Def: translation} and the proofs for $B_1$ and $B_2$. Similarly, the proofs for $C = C_1 \circ C_2$, where $\circ \in \{\wedge, \vee, *\}$ (and $C =\, !C$ for $\BILL$), are straightforward to construct. The only interesting case is when $B= C \to A$, where $C$ is a conservative formula and $A$ is fully conservative.
By Definition \ref{Def: translation}, $(C \to A)^t= (C^t \to A^t) \wedge (A^s \to C^s)$. Now, use the proofs $\tau_2$ for $C$ and $\pi_1$ for $A$, and set $\tau_1$ for $C \to A$ as the following:
\begin{prooftree}
\AxiomC{$\tau_2$}
\noLine
\UnaryInfC{$C \Rightarrow C^t$}
\AxiomC{$\pi_1$}
\noLine
\UnaryInfC{$A^t \Rightarrow A$}
\BinaryInfC{$C^t \to A^t, C \Rightarrow A$}
\UnaryInfC{$C^t \to A^t \Rightarrow C \to A$}
\UnaryInfC{$(C^t \to A^t) \wedge (A^s \to C^s) \Rightarrow C \to A$}
\end{prooftree}
Similarly, using the proofs $\tau_1$ for $C$, and $\pi_2$ and $\sigma$ for $A$, we can set $\tau_2$ for $C \to A$ as:
\begin{prooftree}

\AxiomC{$\tau_1$}
\noLine
\UnaryInfC{$C^t \Rightarrow C$}

\AxiomC{$\pi_2$}
\noLine
\UnaryInfC{$A \Rightarrow A^t$}
\BinaryInfC{$ C \to A, C^t \Rightarrow A^t$}
\UnaryInfC{$ C \to A \Rightarrow C^t \to A^t$}

\AxiomC{$\sigma$}
\noLine
\UnaryInfC{$A^s \Rightarrow \bot$}
\AxiomC{$ C \to A, \bot \Rightarrow C^s$}
\BinaryInfC{$ C \to A, A^s \Rightarrow C^s$}
\UnaryInfC{$ C \to A \Rightarrow A^s \to C^s$}

\BinaryInfC{$ C \to A \Rightarrow (C^t \to A^t) \wedge (A^s \to C^s)$}
\end{prooftree}
Finally, by Definition \ref{Def: translation}, we have $(C \to A)^s=C^t * A^s$. Using the proof $\sigma$ for $A$, we can set $\sigma$ for $C \to A$ as the following:
\begin{prooftree}
\AxiomC{$\sigma$}
\noLine
\UnaryInfC{$A^s \Rightarrow \bot$}
\AxiomC{$ C^t, \bot \Rightarrow \bot$}
\BinaryInfC{$C^t, A^s \Rightarrow \bot$}
\UnaryInfC{$C^t*A^s \Rightarrow \bot$}
\end{prooftree}
For the upper bound on the number of lines, we first prove that $l(\sigma) \leq O(|B|)$. When $B$ is atomic or constant, the corresponding proof consists of a single axiom and hence has only constantly many lines. At each inductive step, we add only constantly many new formulas. Therefore, the total number of proof lines is $O(|B|)$.

For the bounds $l(\pi_i) \leq O(|B|^2)$ and $l(\tau_i) \leq O(|C|^2)$, where $i \in \{1, 2\}$, again, when $B$ or $C$ is atomic or constant, the corresponding proofs have only constantly many lines. At each inductive step, we add only constantly many new formulas to the proofs for the immediate subformulas, except in the case where $B = C \to A$, where we also require the proof $\sigma$ of $A^s \Rightarrow \bot$. Since the latter has $O(|A|)$ many lines, each such step contributes $O(|B|)$ many lines. Hence, in total, $\pi_i$ has $O(|B|^2)$ many lines.
\end{proof}

The following theorem establishes the promised form of feasible conservativity with respect to conservative formulas.

\begin{theorem}\label{Thm: Conservativity}
Let $G \in \{\BMALL, \BAMALL, \BCLL, \BALL\}$. Then, for any conservative formula $A$ and any $G$-proof $\pi$ of $A$ there is an $iG$-proof $\sigma$ of $A$ such that $l(\sigma) \leq |\pi|^{O(1)}$.
\end{theorem}

\begin{proof}
We prove the case $G=\BMALL$; the remaining cases are similar. Set $D=N=\bot$, and let $\pi$ be a $G$-proof of $(\, \Rightarrow A)$, where $A$ is a conservative formula. Since $\MALL$-$\F$ feasibly simulates $\BMALL$ by Lemma \ref{Lem: Equivalence Of Frege}, there exists a $\MALL$-$\F$-proof $\pi'$ of $A$ with $|\pi'|\leq |\pi|^{O(1)}$. By Theorem~\ref{thm: CFLe to FL_e}, there is an $iG_D$-proof $\sigma$ of $(\, \Rightarrow A^t)$ with $l(\sigma)\leq O(|\pi'|)\leq |\pi|^{O(1)}$. As $D=N=\bot$, the axiom $p,N\Rightarrow D$ is provable in $\BIMALL$ with constantly many lines; hence, without loss of generality, we can assume that $\sigma$ is an $\BIMALL$-proof. By Lemma~\ref{lem: basic conservative}, there is an $\BIMALL$-proof $\rho$ of $A^t \Rightarrow A$ such that $l(\rho) \leq O(|A|^2)$. Since $|A| \leq |\pi|$, we obtain an $\BIMALL$-proof of $(\, \Rightarrow A)$ with $|\pi|^{O(1)}$ many lines.
\end{proof}

\begin{remark}
Here are two remarks. First, note the role of \emph{feasibility} in Theorem~\ref{Thm: Conservativity}. Relaxing the feasibility condition, there are some standard conservativity results of the classical systems over their corresponding intuitionistic version (see \cite{kanovich2019, schellinx}). These are typically proved via cut elimination, which leads to an elementary blow-up in proof size and is therefore not suitable for establishing feasible conservativity. The main advantage of Chu's translation is that it allows us to prove conservativity without relying on cut elimination, since the translation handles the cut rule directly.
Second, in proof complexity, there are usually two essentially equivalent ways to handle large formulas in proofs: either by considering the number of formulas in the proof (i.e., the number of lines) instead of the total size, or by allowing extension variables, i.e., new atomic formulas introduced to represent large formulas. We adopt the former for its relative simplicity, though the latter approach is also possible.
\end{remark}

Theorem~\ref{Thm: Conservativity} allows us to transfer exponential lower bounds on the number of proof lines from the intuitionistic setting to lower bounds on proof size in the classical setting.

\begin{theorem} \label{Thm: main}
The $\BFL$-provable sequents $(\Rightarrow \Theta^{*}_n)$ have $\LKu$-proofs of size $n^{O(1)}$, while any of their $\BALL$-proofs has size at least $2^{n^{\Omega(1)}}$. The same result holds when replacing $\LKu$ and $\BALL$ with the Frege systems for $\CPC$ and $\ALL$, respectively.
\end{theorem}

\begin{proof}
By Corollary \ref{RemarkLJAjab}, we know that $\Theta^*_n$ has an $\LKu$-proof of size $n^{O(1)}$. For the exponential lower bound, first recall that 
\[
\Theta^\st_n=\bigast_{i,j} A_{i,j} \to \big((\bigast_{i,l} M_{i,l} \to \alpha_n) \vee (\bigast_{i,l} N_{i,l} \to \beta_n) \big)
\]
where 
\[
A_{i,j}:= (p_{i,j} \wedge 1) \vee (q_{i,j} \wedge 1) \qquad M_{i,l}:= (s_{i,l} \wedge 1) \vee (s'_{i,l} \wedge 1)
\]
\[
N_{i,l}:= (r_{i,l} \wedge 1) \vee  (r'_{i,l} \wedge 1) 
\]
and
\[
\alpha_n:= \alpha^{\lfloor \sqrt{n} \rfloor}_n (\bar{p}, \bar{s}, \bar{s'}) \qquad \beta_n:= \beta^{\lfloor \sqrt{n} \rfloor+1}_n (\bar{q}, \bar{r}, \bar{r'})
\]
where
\[
\alpha^k_n (\bar{p}, \bar{s}, \bar{s'}) := \bigvee_{i\in [n]} \bigwedge_{l \in [k]} s'_{i,l} \;\;\vee\;\; \bigvee_{i\neq j \in [n]} \bigvee_{l \in [k]} (s_{i,l} \wedge s_{j,l} \wedge p_{i,j}),
\]
\[
\beta^k_n (\bar{q}, \bar{r}, \bar{r'}) := \bigvee_{l \in [k]} \bigwedge_{i \in [n]} r'_{i,l} \;\;\vee\;\; \bigvee_{i \neq j \in [n]} \bigvee_{l \neq m \in [k]} (r_{i,l} \wedge r_{j,m} \wedge q_{i,j}).
\]
By Definition \ref{Def: Conservative formulas}, it is clear that all formulas
$A_{i,j}$, $M_{i,l}$, $ N_{i,l}$, $\alpha_n$, and $\beta_n$ are fully conservative and hence conservative. Therefore, the formulas
$\bigast_{i,j} A_{i,j}$, $\bigast_{i,l} M_{i,l}$, and $ \bigast_{i,l} N_{i,l}$ are all conservative.
Thus,
\[
(\bigast_{i,l} M_{i,l} \to \alpha_n) \vee (\bigast_{i,l} N_{i,l} \to \beta_n)
\]
is fully conservative, which implies that the formula $\Theta^\st_n$ is fully conservative and hence conservative.

Now, assume $\BALL \vdash^{\pi_n}\,  \Rightarrow \Theta^\st_n$. As $\Theta^\st_n$  is conservative, by Theorem \ref{Thm: Conservativity}, $\BAILL \vdash^{\sigma_n} \, \Rightarrow \Theta^\st_n$, where $l(\sigma_n) \leq |\pi_n|^{O(1)}$. As $\LJ_!$ is an extension of $\BAILL$, we can say that $\sigma_n$ is also a proof in $\LJ_!$. By Corollary \ref{RemarkLJAjab}, we obtain $|\pi_n|^{O(1)} \geq l(\sigma_n) \geq 2^{n^{\Omega(1)}}$. 
Thus, $|\pi_n| \geq 2^{n^{\Omega(1)}}$. For the claim concerning Frege systems, note that by Lemma \ref{Lem: Equivalence Of Frege}, the Frege systems and sequent calculi for $\CPC$ (resp. for $\ALL$) are equivalent.
\end{proof}
\begin{remark}
Note that Theorem~\ref{Thm: main} implies that any $G$-proof of $(\, \Rightarrow \Theta^*_n)$ has size at least $2^{n^{\Omega(1)}}$, where $G$ is any of the sequent calculi $\BCFL$, $\BCFLw$, $\BMALL$, $\BAMALL$, or $\BCLL$, since they are all subsystems of $\BALL$. The same lower bound also holds for proofs of the formula $\Theta^*_n$ in the corresponding Frege systems for the same reason.
\end{remark}

\section{Cut-free Calculi} \label{Sec: cut-free}
In this section, we assess the strength of the cut rule. We show that removing cut from the already powerful system $\LKajab$ leads to a substantial loss of power: the cut-free system $\mathbf{LK^-_!}$ is exponentially weaker than $\BFL$ for $\mathbf{FL_e}$-provable sequents. More precisely, we provide a family of $\mathbf{FL_e}$-provable sequents that have polynomial-size proofs in $\BFL$ but require exponential-size proofs in $\mathbf{LK^-_!}$. Thus, even with weakening and contraction, the absence of cut cannot be compensated.  

Our strategy is to build an extension $G$ of $\mathbf{FL_e}$ with simple initial sequents to feasibly mimic classical reasoning in $\LK$. Using this connection, we import the known short $\mathbf{LK}$-proof of the Clique–Color formula into $G$. Then, by transforming the initial sequents of $G$ into formulas in the antecedent in the usual way, we obtain a version of the Clique–Color formula with short proofs in $\mathbf{FL_e}$. Finally, we use monotone feasible interpolation for $\mathbf{LK}^{-}$ to show that any short $\mathbf{LK}_!^{-}$-proof of these sequents would yield a small interpolating monotone circuit for the Clique–Color formula, which is impossible.

We begin by introducing formulas in negation normal form in both classical and substructural settings.

\begin{definition}\label{Def: Lnn}
    An $\mathcal{L}_p$-formula $F$ is in \emph{negation normal form} (\emph{nnf}, in short), if it is defined by the following grammar:
\[
F:= p \mid \neg p \mid F_1 \wedge F_2 \mid F_1 \vee F_2
\]
Similarly, an $\calL_u$-formula $F$ is in \emph{$*$-negation normal form} (\emph{$*$-nnf}, in short), if it is defined by the grammar:
\[
F:= p \mid \neg p \mid F_1 \wedge F_2 \mid F_1 \vee F_2 \mid  F_1 \st F_2
\]
We write $\calL_{nn}$ and $\calL^*_{nn}$ for the sets of nnf- and $*$-nnf-formulas.
\end{definition}

As promised, we introduce an extension of the calculus $\mathbf{FL_e}$ by adding a collection of initial sequents that allow us to mimic the classical reasoning (specifically, the weakening and contraction rules) on nnf-formulas within the substructural setting.

\begin{definition}\label{Def: Calculus G}
Define the single-conclusion sequent calculus $G$ over $\calL_u$ as $\BFL$ plus the following initial sequents:
\begin{center}
$\Rightarrow p \vee \neg p \qquad p \Rightarrow 1 \qquad \neg p \Rightarrow 1 \qquad 0 \Rightarrow p
\qquad
0 \Rightarrow \neg p \qquad  0 \Rightarrow 0*0$,
\end{center}
where $p$ is an atom. Note that these initial sequents are not meant to be used as axioms and hence using their substitutions are \emph{not} allowed in $G$. Moreover, $G$ includes the cut rule.
\end{definition}

We begin with some basic observations about the strength of the calculus $G$. The following lemmas show that, over $*$-nnf formulas, $G$ admits short \emph{tree-like} proofs of axiomatic forms of contraction and weakening, as well as the law of excluded middle. These proofs will be used to simulate classical reasoning within $G$. The emphasis on being tree-like is to facilitate the proof of Lemma~\ref{Thm: G to FLe}, as we will explain in Remark \ref{rem: treelike}.

\begin{lemma}\label{lem: law of excluded middle}
For any $A \in \calL^*_{nn}$, the following sequents have tree-like $G$-proofs of size $|A|^{O(1)}$:
\[
A \Rightarrow 1 \qquad 0 \Rightarrow A \qquad \Rightarrow A \vee \neg A \qquad A \Rightarrow A * A.
\]
\end{lemma}

\begin{proof}
The first three statements are proved by induction on the structure of $A$. We begin with the first two sequents. We first present the proofs and then discuss the complexity at the end.

If $A$ is $p$ or $\neg p$, both $A \Rightarrow 1$ and $0 \Rightarrow A$ are present in $G$ as the additional initial sequents.
For the inductive step, let $A = B \circ C$ with $\circ \in \{\wedge, \vee, *\}$. We only treat the case $\circ = *$, as the other cases are similar. For $B * C \Rightarrow 1$, consider the following proof-tree in $G$:
\begin{prooftree}
\AxiomC{$\sigma_1$}
\noLine
\UnaryInfC{$B \Rightarrow 1$}
\AxiomC{$\sigma_2$}
\noLine
\UnaryInfC{$C \Rightarrow 1$}
\UnaryInfC{$C,1 \Rightarrow 1$}
\BinaryInfC{$B , C \Rightarrow 1$}
\UnaryInfC{$B * C \Rightarrow 1$}
\end{prooftree}
where $\sigma_1$ and $\sigma_2$ are provided by the induction hypothesis. 
For the sequent $0 \Rightarrow B * C$, consider the following proof-tree in $G$:
\begin{prooftree}
\AxiomC{$0 \Rightarrow 0*0$}
\AxiomC{$\pi_1$}
\noLine
\UnaryInfC{$0 \Rightarrow B$}   
\AxiomC{$\pi_2$}
\noLine
\UnaryInfC{$0 \Rightarrow C$} 
\BinaryInfC{$0,0 \Rightarrow B*C$}
\UnaryInfC{$0*0 \Rightarrow B*C$}
\BinaryInfC{$0 \Rightarrow B * C$}
\end{prooftree}
where $\pi_1$ and $\pi_2$ are provided by the induction hypothesis. For the size upper bounds, observe that when $A=p$ or $A=\neg p$, the proofs have size $O(|A|)$. In each inductive step toward the sequents $(A \Rightarrow 1)$ and $(0 \Rightarrow A)$, the proof size increases by $O(|A|)$. Since there are $O(|A|)$ connectives to unfold, the resulting $G$-proofs have size $O(|A|^2)$.

For the third sequent, namely $(\, \Rightarrow A \vee \neg A)$, we proceed by induction on the structure of $A$, using the $G$-proofs of $B \Rightarrow 1$, for subformulas $B$ of $A$. As before, we first construct the proof and then establish the upper bound. Recall that $\neg A := A \to 0$, and that the implication rules are available in $\BFL$ and hence in $G$. When $A=p$, the initial sequent $(\, \Rightarrow p \vee \neg p)$ is already included in $G$. When $A=\neg p$, consider the following proof-tree in $G$:
\begin{prooftree}
\AxiomC{$\Rightarrow p \vee \neg p$}
\AxiomC{$p \Rightarrow p$}
\AxiomC{$0 \Rightarrow 0$}
\BinaryInfC{$p, \neg p \Rightarrow 0$}
\UnaryInfC{$p \Rightarrow \neg \neg p$}
\UnaryInfC{$p \Rightarrow \neg  p \vee \neg \neg p$}
\AxiomC{$p \Rightarrow p$}
\AxiomC{$0 \Rightarrow 0$}
\BinaryInfC{$p, \neg p \Rightarrow 0$}
\UnaryInfC{$\neg p \Rightarrow \neg  p$}
\UnaryInfC{$\neg p \Rightarrow \neg  p \vee \neg \neg p$}
\BinaryInfC{$p \vee \neg p \Rightarrow \neg p \vee \neg \neg p$}

\BinaryInfC{$\Rightarrow \neg p \vee \neg \neg p$}
\end{prooftree}
For $A = B \circ C$, where $\circ \in \{\wedge, \vee, *\}$, we obtain a $G$-proof of $\, \Rightarrow A \vee \neg A$ using the induction hypothesis for $B$ and $C$. We explain the case $\circ = \wedge$; the others are similar.  
For $A = B \wedge C$, we first provide $G$-proofs of $X, Y \Rightarrow A \vee \neg A$, for any $X \in \{B, \neg B\}$ and $Y \in \{C, \neg C\}$.  
For $B, C \Rightarrow A \vee \neg A$, consider the following proof-tree in $G$:

\begin{prooftree}
\AxiomC{$\pi_1$}
\noLine
\UnaryInfC{$C \Rightarrow 1$}
\AxiomC{$B \Rightarrow B$}
\UnaryInfC{$B , 1 \Rightarrow B$}
\BinaryInfC{$B, C \Rightarrow B$}
\AxiomC{$\pi_2$}
\noLine
\UnaryInfC{$B \Rightarrow 1$}
\AxiomC{$C \Rightarrow C$}
\UnaryInfC{$C , 1 \Rightarrow C$}
\BinaryInfC{$B, C \Rightarrow C$}
\BinaryInfC{$B, C \Rightarrow B \wedge C$}
\UnaryInfC{$B, C \Rightarrow (B \wedge C) \vee \neg (B \vee C)$}
\end{prooftree}
where $\pi_1$ and $\pi_2$ are the proofs derived from the previous claims for $B$ and $C$. 
For $B, \neg C \Rightarrow A \vee \neg A$, consider the following proof-tree in $G$:
\begin{prooftree}
\AxiomC{$\pi_1$}
\noLine
\UnaryInfC{$B \Rightarrow 1$}
\AxiomC{$C \Rightarrow C$}
\AxiomC{$ 0 \Rightarrow $}
\BinaryInfC{$\neg C, C \Rightarrow $}
\UnaryInfC{$\neg C, C \Rightarrow 0$}
\UnaryInfC{$1, \neg C, C \Rightarrow 0$}
\BinaryInfC{$B, \neg C, C \Rightarrow 0$}
\UnaryInfC{$B, \neg C, (B \wedge C) \Rightarrow 0$}
\UnaryInfC{$B, \neg C \Rightarrow \neg (B \wedge C)$}
\UnaryInfC{$B, \neg C \Rightarrow (B \wedge C) \vee \neg (B \wedge C)$}
\end{prooftree}
Similarly, one can easily find analogous proof-trees in $G$ for the sequents:
\[
\neg B, C \Rightarrow A \vee \neg A \qquad \neg B , \neg C \Rightarrow A \vee \neg A. 
\]
Putting all four proof-trees together and using the rule $(L \vee)$, we get a proof-tree in $G$ for the sequent
\[
B \vee \neg B, C \vee \neg C \Rightarrow (B \wedge C) \vee \neg (B \vee C). 
\]
Finally, by the induction hypothesis, there are proof-trees in $G$ for $\Rightarrow B \vee \neg B$ and $\Rightarrow C \vee \neg C$. Therefore, by applying the cut rule, we obtain a proof-tree for $\Rightarrow (B \wedge C) \vee \neg (B \vee C)$ in $G$.

For the size upper bounds, note that in the cases $A = p$ and $A = \neg p$, the proofs have size $O(|A|)$. In each inductive step establishing the sequent $\Rightarrow A \vee \neg A$, we use the proof-trees of $B \Rightarrow 1$ and $C \Rightarrow 1$, combined with proofs of size $O(|A|)$. Hence, each step contributes $O(|A|^2)$. Since there are $O(|A|)$ connectives to unfold, the total size of the resulting proof-tree is $O(|A|^3)$.

For the last claim, namely $A \Rightarrow A * A$, as $A * A$ is also $*$-nnf, we can consider the following proof in $G$:
\begin{prooftree}
\AxiomC{$A \Rightarrow A$}
\AxiomC{$A \Rightarrow A$}
\BinaryInfC{$A,A \Rightarrow A * A$}
\AxiomC{$A \Rightarrow A$}
\AxiomC{$\pi$}
\noLine
\UnaryInfC{$0 \Rightarrow A*A$}
\BinaryInfC{$A , \neg A \Rightarrow A *A$}
\BinaryInfC{$A \vee \neg A, A \Rightarrow A * A$}
\end{prooftree}
where the proof $\pi$ is obtained in the previous parts. We also showed that there is a proof-tree of $\Rightarrow A \vee \neg A$ in $G$. Applying the cut rule to the conclusions of these two proofs, we obtain a proof-tree for $A \Rightarrow A * A$ in $G$. It is also clear that this proof has size $O(|A|^3)$.
\end{proof}

Now, to embed classical reasoning inside $G$, we consider a sequent calculus operating over nnf $\mathcal{L}_p$-formulas.

\begin{definition}\label{Def: Calculus LKnn}
Define $\LKnn$ over $\calL_{nn}$ as the \emph{single-conclusion} sequent calculus consisting of the following axioms:
\begin{center}
$p \Rightarrow p \qquad \neg p \Rightarrow \neg p \qquad p, \neg p \Rightarrow \qquad \Rightarrow p \vee \neg p$
\end{center}
along with the structural rules (including cut), conjunction and disjunction
rules of $\LJ$. 
\end{definition}

\begin{remark}
The usual calculus $\LKnn$ is presented as a multi-conclusion sequent calculus and includes the axiom $\Rightarrow p, \neg p$. Here, we modify the standard presentation to make the embedding of $\LKnn$ into the single-conclusion system $G$ easier to present.
\end{remark}

It is known that (tree-like) $\LKnn$ can polynomially simulate (tree-like) $\LK$ for single-conclusion sequents of nnf-formulas \cite[Section 3.3]{Krajicek}. Therefore, to simulate $\LK$ in $G$, it is sufficient to do so for $\LKnn$, as long as we only consider nnf $\mathcal{L}_p$-formulas.

\begin{lemma}\label{Lem: LKnn to G}
Let $\Gamma$ and $\Delta$ be multisets of nnf-formulas, where $\Delta$ contains at most one formula. For any tree-like $\LKnn$-proof $\pi$ of $\Gamma \Rightarrow \Delta$, there is a tree-like $G$-proof $\sigma_\pi$ of $\Gamma \Rightarrow \Delta$ such that $|\sigma_\pi| \leq |\pi|^{O(1)}$.  
\end{lemma}
\begin{proof}
The proof proceeds by induction on the structure of the proof $\pi$. We will explain the upper bound at the end. If $\Gamma \Rightarrow \Delta$ is an axiom in $\LKnn$, then clearly $\Gamma \Rightarrow \Delta$ has a proof of linear size in $G$. For the inductive step, since all the rules of $\LKnn$ except weakening and contraction are present in $G$, the claim is a clear consequence of the induction hypothesis. For the structural rules, first note that any nnf-formula in $\mathcal{L}_p$ is also a $*$-nnf. Therefore, by Lemma \ref{lem: law of excluded middle}, for any nnf-formula $A$, we have proof-trees of $A \Rightarrow 1$, $0 \Rightarrow A$ and $A \Rightarrow A*A$ in $G$ of size $|A|^{O(1)} \leq |\pi|^{O(1)}$. Using cut with these sequents, it is easy to simulate rules $(Lw)$, $(Rw)$, and $(Lc)$ of $\LJ$ inside $G$.

Finally, for the upper bound, note that for the axioms, the proof is of size $O(|\pi|)$. In the inductive step, in each step, we add at most $|\pi|^{O(1)}$ many new symbols. Therefore, the size of the constructed proof tree in $G$ is $|\pi|^{O(1)}$.
\end{proof}

\begin{corollary} \label{Cor: LK to G}
Let $\Gamma$ and $\Delta$ be multisets of nnf-formulas, where $\Delta$ contains at most one formula. For any tree-like $\LK$-proof $\pi$ of $\Gamma \Rightarrow \Delta$, there exists a tree-like $G$-proof $\sigma_\pi$ of $\Gamma \Rightarrow \Delta$ such that $|\sigma_\pi| \leq |\pi|^{O(1)}$.
\end{corollary}

\begin{proof}
The tree-like version of the calculi $\LK$ and $\LKnn$ are equivalent over single-conclusion sequents consisting of negation normal form formulas \cite[Section 3.3]{Krajicek}. Therefore, it is enough to apply Lemma \ref{Lem: LKnn to G}.
\end{proof}

An $\calL_u$-formula is called \emph{single-variable} if it contains at most one propositional variable. Define the forgetful translation $f: \mathcal{L}_! \to \mathcal{L}_p$
that maps $0$, $1$, and $*$ to $\bot$, $\top$, and $\wedge$, respectively, and ignores the modality $!$. More precisely:
\begin{definition}\label{Def: forgetful}
Define the \emph{forgetful} function $(\cdot)^f:\calL_! \to \calL_p$ as:
\begin{center}
$p^f:=p \qquad 0^f:= \bot \qquad 1^f:= \top \qquad (A * B)^f:= A^f \wedge B^f$    
\end{center}
\begin{center}
$(A \circ B)^f:= A^f \circ B^f$ for $\circ \in \{\wedge, \vee, \to\}$ \qquad $(!A)^f:=A^f$
\end{center}
Define $\Gamma^f:= \{A^f \mid A \in \Gamma\}$ and $S^f:= \Gamma^f \Rightarrow \Delta^f$, for any multisets $\Gamma$ and $\Delta$ and any sequent $S$ of the form $\Gamma \Rightarrow \Delta$.
\end{definition}

The following lemma is simply a feasible version of the deduction theorem: it feasibly transforms any \emph{tree-like} $G$-proof of a sequent $S$ into an $\BFL$-proof of $S$, at the cost of adding some single-variable classically valid formulas to the antecedent. 

\begin{lemma}[Feasible Deduction Theorem]\label{Thm: G to FLe}
Let $\Gamma$ and $\Delta$ be multisets of $\calL_u$-formulas, where $\Delta$ contains at most one formula. For any tree-like $G$-proof $\pi$ of $\Gamma \Rightarrow \Delta$, there exists a multiset $\Sigma_{\pi}$ of single-variable $\calL_u$-formulas and an $\BFL$-proof $\sigma_\pi$ of $\Sigma_\pi, \Gamma \Rightarrow \Delta$, such that $|\Sigma_\pi| \leq O(|\pi|)$, $|\sigma_\pi| \leq O(|\pi|^3)$, and $\LK \vdash \, \Rightarrow \bigwedge \Sigma_\pi^f$.
\end{lemma}
\begin{proof}
The proof is straightforward and hence deferred to Section~\ref{Sec: Detailed Proofs}.
\end{proof}

\begin{remark}\label{rem: treelike}
In the proof of Lemma \ref{Thm: G to FLe}, the multiset $\Sigma_{\pi}$ is constructed by recursion on the structure of $\pi$. When the last rule in $\pi$ is a binary rule with immediate subtrees $\pi'$ and $\pi''$, we define $\Sigma_{\pi} = \Sigma_{\pi'} \cup \Sigma_{\pi''}$ (multiset union). To ensure that the size of $\Sigma_{\pi}$ remains polynomial in the size of $\pi$, it is important that $\pi$ is tree-like. Otherwise, it is possible that $\pi'=\pi''$ and $\pi$ considers them once. Hence, passing from $\pi'$ to $\pi$ takes one step while it makes the size of $\Sigma_{\pi}$ twice the size of $\Sigma_{\pi'}$. This may make the size of $\Sigma_{\pi}$ exponential in the size of $\pi$.  
\end{remark}

From Section \ref{sec: HardTheorems}, recall the formulas 
\[
Color_n := Color^{\lfloor\sqrt{n}\rfloor}_n(\bar{p}, \bar{s}) \qquad Clique_n := Clique^{\lfloor\sqrt{n}\rfloor+1}_n(\bar{p}, \bar{r}),
\]
and note that they are in negation normal form. Let $\overline{Color}_n(\bar{p}, \bar{s})$ denote the formula $\neg Color_n(\bar{p}, \bar{s})$ in negation normal form.
From Section \ref{sec: HardTheorems}, there exists a sequence of tree-like $\LK$-proofs $\pi_n$ for each sequent
\[
Clique_n(\bar{p}, \bar{r}) \Rightarrow \overline{Color}_n(\bar{p}, \bar{s})
\]
such that $|\pi_n|$ is polynomial in $n$. We can now transform these sequents into $\BFL$-provable sequents with short proofs.
\begin{corollary} \label{Cor: LK to FLe}
There exist a sequence $\{\Sigma_n\}_{n=1}^{\infty}$ of multisets of single-variable $\calL_u$-formulas and a sequence $\{\sigma_n\}_{n=1}^{\infty}$ such that each $\sigma_n$ is an $\BFL$-proof for
\[
\Sigma_{n},\; Clique_n(\bar{p}, \bar{r}) \Rightarrow \overline{Color}_n(\bar{p}, \bar{s}),
\]
$|\Sigma_{n}|,|\sigma_{n}| \leq n^{O(1)}$ and $\LK \vdash \, \Rightarrow \bigwedge \Sigma_{n}^f$.
\end{corollary}

\begin{proof} 
Use Theorem \ref{Thm: Buss Clique Color}, Corollary \ref{Cor: LK to G} and Lemma \ref{Thm: G to FLe}.
\end{proof}

Consider the atoms appearing in the Clique–Color formulas, namely $\bar{p}$, $\bar{r}$, and $\bar{s}$. Take the multiset $\Sigma_{n}$ derived in Corollary \ref{Cor: LK to FLe} and define the multiset $\Sigma_{n}^{\bar{p}, \bar{s}}$ as
\[
\Sigma_{n}^{\bar{p}, \bar{s}} := \{A \in \Sigma_{n} \mid V(A) \subseteq \bar{p} \cup \bar{s}\}.
\]
Note that the multiplicity of any $A \in \Sigma^{\bar{p}, \bar{s}}_{n}$ is the same as its multiplicity in $\Sigma_{n}$.
Define $\Pi_n$ as the remainder of $\Sigma_n$, namely $\Pi_n:= \Sigma_n \setminus \Sigma_{n}^{\bar{p}, \bar{s}}$. Since the formulas in $\Sigma_n$ are single-variable, the formulas in $\Pi_n$ contain no occurrences of the variables in $\bar{p} \cup \bar{s}$.

The following theorem is the main theorem of this section. We provide a sequence of sequents that have short proofs even without the contraction or weakening rules in the intuitionistic setting, namely in $\BFL$. However, interestingly, if we remove the cut rule from the sequent calculus and even add both contraction and weakening rules, namely in $\mathbf{LK^-_!}$, any proof of this sequence of sequents blows up exponentially.

\begin{theorem}\label{MainCutFree}
The sequents
\[
S_n := Clique_n(\bar{p}, \bar{r}),\; \Pi_n \;\Rightarrow\; \bigast \Sigma_n^{\bar{p}, \bar{s}} \to \overline{Color}_n(\bar{p}, \bar{s})
\]
have $\BFL$-proofs of size $n^{O(1)}$, but any $\mathbf{LK^-_!}$-proof of $S_n$ has size $2^{n^{\Omega(1)}}$.
\end{theorem}
\begin{proof}
By Corollary \ref{Cor: LK to FLe}, we obtain an $\mathbf{FL_e}$-proof $\sigma_n$ of
\[
\Sigma_{n},\; Clique_n(\bar{p}, \bar{r}) \Rightarrow \overline{Color}_n(\bar{p}, \bar{s}),
\]
where $|\Sigma_{n}|$ and $|\sigma_{n}|$ are both polynomial in $n$.  
Since $\Sigma_n = \Pi_n \cup \Sigma_n^{\bar{p}, \bar{s}}$, using the rules $(L*)$ and $(R\to)$, we can construct an $\BFL$-proof of $S_n$ whose size is polynomial in $n$.

Now, we show that any $\mathbf{LK^-_!}$-proof of $S_n$ has size $2^{n^{\Omega(1)}}$. Recall that $Clique_n(\bar{p}, \bar{r})$ is monotone in $\bar{p}$. As there is no occurrence of $\bar{p}$ in $\Pi_n$, the antecedent of $S_n$ is also monotone in $\bar{p}$. Moreover, since the only variables in the succedent of $S_n$ are $\bar{p}$ and $\bar{s}$, and none of these variables occur in $\Pi_n$, the only atoms shared by the antecedent and the succedent of $S_n$ are those in $\bar{p}$.
Let $\{\rho_{n}\}_n$ be a sequence of $\mathbf{LK^-_!}$-proofs of $S_n$. As $\LK^-$ polynomially simulates $\LK^-_!$ up to the translation $f$, we obtain $\LK^-$-proofs $\tau_n$ for $S_n^f$ such that $|\tau_n| \leq |\rho_n|^{O(1)}$.
Using the monotone feasible interpolation property of $\LKminus$ on $\tau_n$'s, we get a monotone circuit $C_n$ in the inputs in $\bar{p}$ with size bounded by $|\tau_n|^{O(1)} \leq |\rho_n|^{O(1)}$ to compute a Craig interpolant $I_n$ for $S_n^f$. Any interpolant $I_n$ of $S_n^f$ is clearly an interpolant for
\[
Clique_n(\bar{p}, \bar{r}) \Rightarrow  \neg Color_n(\bar{p}, \bar{s}),
\]
as any formula in $\Sigma_n^f$ is classically valid.
By Theorem \ref{Thm: Alon}, the above sequent requires monotone interpolating circuit of size $2^{\Omega(n^{1/4})}$. Hence, $|\rho_n|^{O(1)} \geq |C_n| \geq 2^{n^{\Omega(1)}}$. Therefore, $|\rho_n| \geq 2^{n^{\Omega(1)}}$.
\end{proof}

Consequently, we get the following exponential speed-up between calculi for several linear logics and their cut-free versions.

\begin{corollary}\label{Cor: cut free}
There is an exponential speed-up between any of the sequent calculi $\BFL$, $\BFLc$, $\BFLw$, $\LJ$, $\BCFL$, $\BCFLc$, $\BCFLw$, $\LK$, $\BIMALL$, $\BAIMALL$, $\BRIMALL$, $\BILL$, $\BCLL$, $\BELL$, $\BALL$, $\BRLL$   and their cut-free versions.
\end{corollary}
\begin{proof}
The reason is simply that all these systems extend $\mathbf{FL}_e$ on the one hand, and their cut-free versions can be polynomially simulated by $\LK^{-}_{!}$.
\end{proof}

\section{Conclusion and Future Work}
We have isolated the power of structural rules in producing efficient proofs in the classical sequent calculus $\LK$ by systematically removing contraction and cut. Eliminating any single rule produces exponential proof-size blow-ups, yielding new separation results and lower bounds for substructural and linear proof systems. In particular, this answers positively the question posed in \cite{Raheleh} regarding extending the exponential lower bound for $\mathbf{CFL_{ew}^-}$ to sequent calculi with cut. A key observation is that these separations are not sensitive to the addition of the controlled structural reasoning supported by the linear exponential ``$!$''. This shows that the obtained hardness is intrinsic to the absence of structural rules and cannot be compensated by controlled versions of these rules.

For future work, Table \ref{table: summary} highlights remaining challenges, such as establishing exponential separations between logics without weakening and contraction versus those with only weakening or contraction, as well as comparisons among cut-free calculi.
Another promising direction is asking whether exponential lower bounds can be proved for the number of proof lines rather than overall proof size for affine linear logic. This is closely related to extending our results from Frege systems to extended Frege (EF) systems, which formalize the practice of naming large subformulas via extension axioms to obtain shorter proofs. EF systems are more robust: bounds on EF proof size essentially correspond to bounds on the number of lines in a Frege proof, and they can be viewed as Frege systems operating with circuits, allowing many arguments to carry over more smoothly. Inspired by Question 9.2 in \cite{Jerabek2}, it is also natural to ask whether Frege and extended Frege systems can be separated in the substructural and linear setting.\\

\noindent \textbf{Acknowledgment.} We would like to thank Nick Galatos, Erfan Khaniki, and Vitor Greati for fruitful discussions. A. Akbar Tabatabai acknowledges the financial support of the Dutch Research Council (NWO) through project OCENW.M.22.258.

\section{Detailed Proofs} \label{Sec: Detailed Proofs}

This section contains some detailed proofs deferred from previous sections to improve the readability of the main arguments. For the reader's convenience, we restate the corresponding lemmas and theorems.



\ThmCFLeToFLe*
\begin{proof}
We begin by proving the first part, where $D$ is arbitrary, and restrict attention to the case $L=\mathsf{LL}$. The remaining cases follow by the same argument, simply by restricting the language. For the second part, it suffices to address the remaining axiom, i.e., weakening, which we treat at the end.

Let $\pi= A_1, \ldots, A_n$ be an $L$-$\F$ proof of $A$ using the axioms and rules in Tables~\ref{Fig: FL_e} and~\ref{Fig: linear}. 
To construct an $iG_D$-proof of $(\, \Rightarrow A^t)$ with $O(|\pi|)$ many lines, it suffices to establish the following two points. First, for any instance $\alpha$ of an axiom of $L$-$\F$, the sequent $(\, \Rightarrow \alpha^t)$ has an $iG_D$-proof with $O(|\alpha|)$ many lines. Second, each rule of $L$-$\F$ can be simulated in $iG_D$ by adding only constantly many lines. We begin by analyzing the axiom case.

$\bullet$ Axiom $(\text{id})$: Let $\alpha$ be an instance of the axiom $(\text{id})$. Then:
\[
\alpha^t = (A \to A)^t= (A^t \to A^t) \wedge (A^s \to A^s).
\]
Clearly, $(\, \Rightarrow \alpha^t)$ has an $iG_D$-proof with $O(1)$ many lines.

$\bullet$ Axiom $(\text{pf})$:
Let $\alpha$ be an instance of the axiom $(\text{pf})$:
\[
(A \rightarrow B) \rightarrow \big((C \rightarrow A) \rightarrow(C \rightarrow B) \big).
\]
Then, by Definition \ref{Def: translation}, $\alpha^t= \beta \wedge \gamma$, where
\[
\beta = (A \to B)^t \to \big((C \rightarrow A) \rightarrow(C \rightarrow B) \big)^t,
\]
\[
\gamma= \big((C \rightarrow A) \rightarrow(C \rightarrow B) \big)^s \to (A \to B)^s.
\]
First, let us construct an $iG_D$-proof for $\beta$. By Definition \ref{Def: translation}, $\beta= \phi \to (\psi \wedge \theta)$, where
\[
\phi = (A^t \to B^t) \wedge (B^s \to A^s),
\]
\[
\psi= [(C^t \to A^t) \wedge (A^s \to C^s)] \to [(C^t \to B^t)
\wedge (B^s \to C^s)], 
\]
\[
\theta= (C^t * B^s \to C^t *A^s).  
\]
Now, since:
\[
iG_D \vdash B^s \to A^s \Rightarrow C^t * B^s \to C^t *A^s,
\]
we get $iG_D \vdash \phi \Rightarrow \theta$. Moreover, as the following sequents are provable in $iG_D$:
\[
C^t \to A^t , A^t \to B^t \Rightarrow C^t \to B^t,
\]
\[
B^s \to A^s, A^s \to C^s \Rightarrow B^s \to C^s,
\]
we get $iG_D \vdash \phi \Rightarrow \psi$. Thus $iG_D \vdash \; \Rightarrow \beta$. Inspecting the construction of the $iG_D$-proof, it is clear that it contains only a constant number of formulas.

For the second conjunct of $\alpha^t$, namely $\gamma$, we have
\[
\gamma= [((C^t \to A^t) \wedge (A^s \to C^s)) * (C^t * B^s)] \rightarrow A^t * B^s. 
\]
Using 
\[
iG_D \vdash C^t \to A^t, C^t, B^s \Rightarrow A^t * B^s,
\]
we clearly get an $iG_D$-proof for $(\, \Rightarrow \gamma)$ with constantly many lines. Thus, the sequent $(\, \Rightarrow \alpha^t)$ has an $iG_D$-proof with $O(1)$ many lines.

$\bullet$ Axiom (per): Let $\alpha$ be an instance of the axiom (per):
\[
\big(A \rightarrow(B \rightarrow C)\big) \rightarrow \big(B \rightarrow(A \rightarrow C) \big).
\]
Then, by Definition \ref{Def: translation}, $\alpha^t= \beta \wedge \gamma$, where
\[
\beta= \big(A \rightarrow(B \rightarrow C)\big)^t \rightarrow \big(B \rightarrow(A \rightarrow C) \big)^t,
\]
\[
\gamma= \big(B \rightarrow(A \rightarrow C) \big)^s \to \big(A \rightarrow(B \rightarrow C)\big)^s.
\]
First, we construct an $iG_D$-proof for $\beta$. By Definition \ref{Def: translation}, $\beta= \phi \to (\psi_1 \wedge \psi_2)$, where
\[
\phi= [A^t \to ((B^t \to C^t) \wedge (C^s \to B^s))] \wedge (B^t * C^s \to A^s),
\]
\[
\psi_1= B^t \to ((A^t \to C^t) \wedge (C^s \to A^s)),
\]
\[
\psi_2= A^t * C^s \to B^s.
\]
Now, to get $iG_D \vdash \phi \Rightarrow \psi_1$,
it is enough to prove $iG_D \vdash \phi, B^t \Rightarrow A^t \to C^t$ and $iG_D \vdash \phi, B^t \Rightarrow C^s \to A^s$. For the first, consider 
\[
iG_D \vdash A^t \to (B^t \to C^t), B^t  \Rightarrow A^t \to C^t
\]
and
\[
iG_D \vdash \phi \Rightarrow  A^t \to (B^t \to C^t)
\]
and apply the cut rule.
For the second, it is enough to use the following:
\[
iG_D \vdash (B^t *C^s) \to A^s, B^t \Rightarrow C^s \to A^s.
\]
Therefore, we have shown $iG_D \vdash \phi \Rightarrow \psi_1$. In each case, it is clear that the constructed $iG_D$-proofs have $O(1)$ many lines.

For $iG_D \vdash \phi \Rightarrow \psi_2$, note that
\[
iG_D \vdash \phi \Rightarrow A^t \to (C^s \to B^s)
\]
and
\[
iG_D \vdash A^t \to (C^s \to B^s) \Rightarrow A^t * C^s \to B^s.
\]
Then, we apply cut. Again, it is clear that the constructed $iG_D$-proof has $O(1)$ many lines.

For the second conjunct of $\alpha^t$, namely $\gamma$, we have
\[
\gamma = B^t *(A^t * C^s) \to A^t * (B^t * C^s),
\]
which clearly has an $iG_D$-proof with constant number of lines 
by the associativity and commutativity of multiplication.

$\bullet$ Axiom $(* \wedge)$: Let $\alpha$ be an instance of the axiom $(* \wedge)$:
\[
\big((A \wedge 1) * (B \wedge 1)\big) \rightarrow (A \wedge B).
\]
Then, by Definition \ref{Def: translation}, $\alpha^t= \beta \wedge \gamma$, where
\[
\beta= \big((A \wedge 1) * (B \wedge 1)\big)^t \to (A \wedge B)^t,
\]
\[
\gamma= (A \wedge B)^s \to \big((A \wedge 1) * (B \wedge 1)\big)^s.
\]
To prove $\beta$, by Definition \ref{Def: translation} we get
\[
\beta= (A^t \wedge 1) * (B^t \wedge 1) \to A^t \wedge B^t. 
\]
Clearly $iG_D \vdash \; \Rightarrow \beta$ with a proof with constant many lines. For $\gamma$, take
\[
\phi= (A^t \wedge 1) \to (B^s \vee D) \qquad
\psi= (B^t \wedge 1) \to (A^s \vee D).
\]
By Definition \ref{Def: translation} we get
\[
\gamma= (A^s \vee B^s) \to (\phi \wedge \psi).
\]
Let us prove $iG_D \vdash A^s \Rightarrow (\phi \wedge \psi)$. 
The case of $iG_D \vdash B^s \Rightarrow (\phi \wedge \psi)$ is similar. By Lemma \ref{Lem: iGD}, the sequent $A^t, A^s \Rightarrow D$ and hence $A^s, A^t \wedge 1 \Rightarrow B^s \vee D$ has an $iG_D$-proof with $O(|A|) \leq O(|\alpha|)$ many lines. Moreover, $A^s, B^t \wedge 1 \Rightarrow A^s \vee D$ clearly has an $iG_D$-proof with $O(1)$ many lines. 
Therefore, using the $iG_D$-proofs of
\[
A^s, A^t \wedge 1 \Rightarrow B^s \vee D \qquad
A^s, B^t \wedge 1 \Rightarrow A^s \vee D,
\]
we get $iG_D \vdash A^s \Rightarrow (\phi \wedge \psi)$ with a proof with $O(|\alpha|)$ many lines.

$\bullet$ Axiom $(\wedge \! \to)_1$: Let $\alpha$ be an instance of the axiom:
\[
(A \wedge B) \rightarrow A.
\]
Then, by Definition \ref{Def: translation}:
\[
\alpha^t= (A^t \wedge B^t \to A^t) \wedge (A^s \to A^s \vee B^s).
\]
Clearly $(\, \Rightarrow \alpha^t)$ has an $iG_D$-proof with $O(1)$ many lines. Similarly, we can prove the case of the axiom $(\wedge \! \to)_2$.

$\bullet$ Axiom $(\to \!\wedge)$: Let $\alpha$ be an instance of the axiom $(\rightarrow\! \wedge)$:
\[
\big((A \rightarrow B) \wedge(A \rightarrow C)\big) \rightarrow \big(A \rightarrow(B \wedge C) \big).
\]
Then, by Definition \ref{Def: translation}, $\alpha^t= \beta \wedge \gamma$, where
\[
\beta= (A \to B)^t \wedge (A \to C)^t \to (A \to B \wedge C)^t,
\]
\[
\gamma= (A \to B \wedge C)^s \to (A \to B)^s \vee (A \to C)^s.
\]
Hence, 
\[
\beta= [((A^t \to B^t) \wedge (B^s \to A^s)) \wedge ((A^t \to C^t) \wedge (C^s \to A^s))] \to 
\]
\[
[(A^t \to B^t \wedge C^t) \wedge (B^s \vee C^s \to A^s)],
\]
\[
\gamma= A^t * (B^s \vee C^s) \to (A^t * B^s) \vee (A^t * C^s).
\]
The sequent $(\, \Rightarrow \beta)$ is provable in $iG_D$ because 
\[
iG_D \vdash (A^t \to B^t) \wedge (A^t \to C^t) \Rightarrow A^t \to (B^t \wedge C^t),
\]
\[
iG_D \vdash (B^s \to A^s) \wedge (C^s \to A^s) \Rightarrow (B^s \vee C^s) \to A^s.
\]
And $(\, \Rightarrow \gamma)$ is provable in $iG_D$ because of the distributivity of $*$ over $\vee$. The $iG_D$-proofs for $(\, \Rightarrow \beta)$ and $(\, \Rightarrow \gamma)$ have $O(1)$ many lines. Thus, $(\, \Rightarrow \alpha^t)$ has an $iG_D$-proof with $O(1)$ many lines.

$\bullet$ Axiom $(\rightarrow \! \vee)_1$: Let $\alpha$ be an instance of the axiom:
\[
A \rightarrow(A \vee B).
\]
Then, by Definition \ref{Def: translation},
\[
\alpha^t = (A^t \to A^t \vee B^t) \wedge (A^s \wedge B^s \to A^s)
\]
Clearly $(\, \Rightarrow \alpha^t)$ has an $iG_D$-proof with $O(1)$ many lines. Similarly, we can prove the case of the axiom $(\to \! \vee)_2$.

$\bullet$ Axiom $(\vee \! \rightarrow)$: Let $\alpha$ be an instance of the axiom $(\vee \! \rightarrow)$:
\[
\big((A \rightarrow C) \wedge(B \rightarrow C)\big) \rightarrow (A \vee B \rightarrow C).
\]
Then, by Definition \ref{Def: translation}, $\alpha^t= \beta \wedge \gamma$, where
\[
\beta= \big((A \rightarrow C) \wedge(B \rightarrow C)\big)^t \rightarrow (A \vee B \rightarrow C)^t,
\]
\[
\gamma=(A \vee B \rightarrow C)^s \to \big((A \rightarrow C) \wedge(B \rightarrow C)\big)^s.
\]
To prove $\beta$, take 
\[
\phi=((A^t \to C^t) \wedge (C^s \to A^s)) \wedge ((B^t \to C^t) \wedge (C^s \to B^s)),
\]
\[
\psi_1= (A^t \vee B^t) \to C^t \qquad \psi_2= C^s \to (A^s \wedge B^s).
\]
Then $\beta= \phi \to (\psi_1 \wedge \psi_2)$. Since the following sequents have $iG_D$-proofs with $O(1)$ many lines:
\[
(A^t \to C^t) \wedge (B^t \to C^t) \Rightarrow (A^t \vee B^t) \to C^t, 
\]
\[
(C^s \to A^s) \wedge (C^s \to B^s) \Rightarrow C^s \to (A^s \wedge B^s),
\]
we get an $iG_D$-proof for $(\, \Rightarrow \beta)$ with $O(1)$ many lines. For $\gamma$, by Definition \ref{Def: translation}:
\[
\gamma= (A^t \vee B^t) * C^s \to (A^t * C^s) \vee (B^t * C^s).
\]
Clearly, $(\; \Rightarrow \gamma)$ has an $iG_D$-proof with $O(1)$ many lines. 

$\bullet$ Axiom $(\rightarrow \! *)$: Let $\alpha$ be an instance of the axiom $(\rightarrow \! *)$:
\[
B \rightarrow(A \rightarrow (A \st B)).
\]
Then, by Definition \ref{Def: translation}, $\alpha^t= \beta \wedge \gamma$, where
\[
\beta= B^t \to (A \to (A *B))^t
\qquad
\gamma= (A \to (A *B))^s \to B^s.
\]
For $\beta$, take 
\[
\phi_1= A^t \to A^t * B^t
\quad
\phi_2= [(A^t \to B^s) \wedge (B^t \to A^s)] \to A^s.
\]
Thus, $\beta= B^t \to \phi_1 \wedge \phi_2$. Clearly the sequents
\[
B^t \Rightarrow (A^t \to A^t * B^t) \qquad B^t, B^t \to A^s \Rightarrow A^s
\]
have $iG_D$-proofs with $O(1)$ many lines and so does $(\, \Rightarrow \beta)$. For $\gamma$ we have
\[
\gamma = \big(A^t * ((A^t \to B^s) \wedge (B^t \to A^s))\big) \to B^s.
\]
As $A^t * (A^t \to B^s) \Rightarrow B^s$ has an $iG_D$-proof with $O(1)$ many lines, so does $(\, \Rightarrow \gamma)$.

$\bullet$ Axiom $(* \! \rightarrow)$: Let $\alpha$ be an instance of the axiom $(* \! \rightarrow)$:
\[
\big(B \rightarrow(A \rightarrow C)\big) \rightarrow((A \st B) \rightarrow C).
\]
Then, by Definition \ref{Def: translation}, $\alpha^t= \beta \wedge \gamma$, where
\[
\beta = \big(B \rightarrow(A \rightarrow C)\big)^t  \rightarrow((A \st B) \rightarrow C)^t,
\]
\[
\gamma = ((A \st B) \rightarrow C)^s \to \big(B \rightarrow(A \rightarrow C)\big)^s.
\]
For $\beta$, take
\[
\phi_1= B^t \to \big((A^t \to C^t) \wedge (C^s \to A^s) \big) \quad \phi_2= A^t * C^s \to B^s
\]
\[
\psi_1= (A^t * B^t) \to C^t \quad \psi_2 = C^s \to \big((A^t \to B^s) \wedge (B^t \to A^s)  \big)
\]
Thus $\beta = (\phi_1 \wedge \phi_2) \to (\psi_1 \wedge \psi_2)$. Now, notice
\begin{itemize}
\item 
$iG_D \vdash B^t \to (A^t \to C^t) \Rightarrow A^t * B^t \to C^t$,
\item 
$iG_D \vdash C^s, A^t, A^t * C^s \to B^s \Rightarrow B^s$,
\item 
$iG_D \vdash C^s, B^t, B^t \to (C^s \to A^s) \Rightarrow A^s$,
\end{itemize}
all with proofs with $O(1)$ many lines. Hence, 
\begin{itemize}
\item 
$iG_D \vdash \phi_1 \Rightarrow \psi_1$,
\item 
$iG_D \vdash \phi_2, C^s \Rightarrow (A^t \to B^s)$,
\item 
$iG_D \vdash \phi_1, C^s \Rightarrow (B^t \to A^s)$,
\end{itemize}
all with proofs with $O(1)$ many lines. Thus, $iG_D \vdash \; \Rightarrow \beta$ with a proof with $O(1)$ many lines. For $\gamma$ we have
\[
\gamma= (A^t * B^t) * C^s \to B^t * (A^t * C^s).
\]
By the commutativity and associativity of multiplication, $(\, \Rightarrow \gamma)$ clearly has an $iG_D$-proof with $O(1)$ many lines. 

$\bullet$ Axiom (1): Clear by Definition \ref{Def: translation}.

$\bullet$ Axiom $(1 \! \rightarrow)$: Let $\alpha$ be an instance of the axiom $(1 \! \rightarrow)$:
\[
1 \rightarrow(A \rightarrow A).
\]
Then, by Definition \ref{Def: translation}, $\alpha^t= \beta \wedge \gamma$, where
\[
\beta= 1 \to (A \to A)^t \qquad
\gamma= (A \to A)^s \to D.
\]
Thus
\[
\beta= 1 \to ((A^t \to A^t) \wedge (A^s \to A^s))
\qquad
\gamma= A^t * A^s \to D.
\]
Therefore, the sequents $(\, \Rightarrow \beta)$ and $(\, \Rightarrow \gamma)$ have $iG_D$-proofs with $O(1)$ and $O(|A|) \leq O(|\alpha|)$ many lines, respectively. In the case of $\gamma$, we use Lemma~\ref{Lem: iGD}.

$\bullet$ Axiom (top): Take an instance $\alpha$ of the axiom $A \to \top$. By Definition:
\[
\alpha^t = (A^t \to \top) \wedge (\bot \to A^s).
\]
Clearly $(\, \Rightarrow \alpha^t)$ has an $iG_D$-proof with $O(1)$ many lines.

$\bullet$ Axiom (bot): Take an instance $\alpha$ of the axiom $\bot \to A$. By Definition:
\[
\alpha^t = (\bot \to A^t) \wedge (A^s \to \top).
\]
Clearly $(\, \Rightarrow \alpha^t)$ has an $iG_D$-proof with $O(1)$ many lines.

$\bullet$ Axiom $(\text{dn})$:
Take an instance $\alpha$ of the axiom $\neg \neg A \to A$. By Definition:
\[
\alpha^t= (\neg \neg A \to A)^t= ((\neg \neg A)^t \to A^t) \wedge (A^s \to (\neg \neg A)^s).
\]
By Lemma \ref{Lem: iGD}, the following sequents have $iG_D$-proofs with $O(|A|) \leq O(|\alpha|)$ many lines:
\[
(\neg \neg A)^t \Leftrightarrow (\neg A)^s \Leftrightarrow A^t \qquad (\neg \neg A)^s \Leftrightarrow (\neg A)^t \Leftrightarrow A^s.
\]
Thus, $(\, \Rightarrow \alpha^t)$ has an $iG_D$-proof with $O(|\alpha|)$ many lines.

$\bullet$ Axiom $(!w)$: Let $\alpha$ be an instance of the axiom $(!w)$:
\[
A \to (!B \to A).
\]
Then, by Definition \ref{Def: translation}, $\alpha^t= \beta \wedge \gamma$, where
\[
\beta= A^t \to (!B \to A)^t
\qquad
\gamma= (!B \to A)^s \to A^s.
\]
For $\beta$, note that
\[
\beta= A^t \to \big((!B^t \to A^t) \wedge (A^s \to (!B^t \to D))  \big).
\]
Using the rule $(W!)$, we clearly have an $iG_D$-proof of $A^t, !B^t \Rightarrow A^t$ with $O(1)$ many lines. Moreover, by Lemma \ref{Lem: iGD} and the rule $(W!)$, the sequent
$A^t \Rightarrow (A^s \to (!B^t \to D))$ has an $iG_D$-proof with $O(|A|) \leq O(|\alpha|)$ many lines. Therefore, $(\, \Rightarrow \beta)$ has an $iG_D$-proof with $O(|\alpha|)$ many lines. For $\gamma$ we have 
\[
\gamma = !B^t * A^s \to A^s.
\]
Clearly $(\, \Rightarrow \gamma)$ has an $iG_D$-proof with $O(1)$ many lines.

$\bullet$ Axiom $(!c)$: Let $\alpha$ be an instance of the axiom $(!c)$:
\[
(!A \to (!A \to B)) \to (!A \to B).
\]
Then, by Definition \ref{Def: translation}, $\alpha^t= \beta \wedge \gamma$, where
\[
\beta= (!A \to (!A \to B))^t \to (!A \to B)^t,
\]
\[
\gamma= (!A \to B)^s \to (!A \to (!A \to B))^s.
\]
For $\beta$, take
\[
\phi_1= !A^t \to [(!A^t \to B^t) \wedge (B^s \to (!A^t \to D))]
\]
\[
\phi_2= (!A^t * B^s )\to (!A^t \to D)
\]
\[
\psi_1= !A^t \to B^t \qquad \psi_2= B^s \to (!A^t \to D)
\]
Thus, $\beta= (\phi_1 \wedge \phi_2) \to (\psi_1 \wedge \psi_2)$. Using the rule $(C!)$, the sequent
\[
!A^t \to (!A^t \to B^t) \Rightarrow (!A^t \to B^t) 
\]
has an $iG_D$-proof with $O(1)$ many lines and so does $\phi_1 \Rightarrow \psi_1$. Moreover, 
\[
!A^t, !A^t, B^s, (!A^t * B^s )\to (!A^t \to D) \Rightarrow D
\]
has an $iG_D$-proof with $O(1)$ many lines.
Thus, using the rule $(C!)$ in $iG_D$, we get an $iG_D$-proof of 
\[
!A^t, B^s, (!A^t * B^s )\to (!A^t \to D) \Rightarrow D
\]
with $O(1)$ many lines.
Thus, $iG_D \vdash \phi_2 \Rightarrow \psi_2$ and hence $iG_D \vdash (\, \Rightarrow \beta)$ has an $iG_D$-proof with $O(1)$ many lines. For $\gamma$ we have 
\[
\gamma= (!A^t * B^s) \to !A^t * (!A^t * B^s).
\]
Using the rule $(C!)$, the sequent $(\, \Rightarrow \gamma)$ clearly has an $iG_D$-proof with $O(1)$ many lines.

$\bullet$ Axiom $(!\text{K})$: Let $\alpha$ be an instance of the axiom $(!\text{K})$:
\[
!(A \to B) \to (!A \to !B).
\]
Then, by Definition \ref{Def: translation}, $\alpha^t= \beta \wedge \gamma$, where
\[
\beta= (!(A \to B))^t \to (!A \to !B)^t,
\]
\[
\gamma= (!A \to !B)^s \to (!(A \to B))^s.
\]
For $\beta$, take
\[
\phi= !((A^t \to B^t) \wedge (B^s \to A^s))
\]
\[
\psi_1=!A^t \to !B^t \qquad \psi_2= (!B^t \to D) \to (!A^t \to D).
\]
Thus, $\beta= \phi \to (\psi_1 \wedge \psi_2)$. To construct an $iG_D$-proof of $\phi \Rightarrow \psi_1$, consider the sequent
\[
A^t \to B^t, A^t \Rightarrow B^t,
\]
that has an $iG_D$-proof with $O(1)$ many lines.
Then, by applying the rules $(L \wedge)$, $(L!)$, and $(R!)$, we obtain an $iG_D$-proof of
\[
S:=!\big((A^t \to B^t) \wedge (B^s \to A^s)\big), !A^t \Rightarrow !B^t,
\]
and hence of $\phi \Rightarrow \psi_1$ with $O(1)$ many lines. For $\phi \Rightarrow \psi_2$, construct an $iG_D$-proof of
\[
!B^t, !B^t \to D \Rightarrow D.
\]
Then, using the constructed $iG_D$-proof of $S$ and applying the cut rule, we get an $iG_D$-proof of
\[
!\big((A^t \to B^t) \wedge (B^s \to A^s)\big), !A^t, !B^t \to D \Rightarrow D,
\]
which easily provides an $iG_D$-proof of $\phi \Rightarrow \psi_2$. It is easy to see that this proof has $O(1)$ many lines.

For $\gamma$ we have 
\[
\gamma = (!A^t * (!B^t \to D)) \to (\phi \to D).
\]
Note that $(\, \Rightarrow \gamma)$ is already proved in the discussion for $\beta$.

$\bullet$ Axiom $(!\text{T})$: Let $\alpha$ be an instance of the axiom $!A \to A$. Then, by Definition \ref{Def: translation},
\[
\alpha^t= (!A^t \to A^t) \wedge (A^s \to (!A^t \to D)).
\]
The left conjunct is clearly provable by the rule $(L!)$ and the right conjunct can be derived using Lemma \ref{Lem: iGD} and the rule $(L!)$. It is clear that the first proof has $O(1)$ many lines, while the second has $O(|A|) \leq O(|\alpha|)$ many lines.

$\bullet$ Axiom $(!\text{4})$: Let $\alpha$ be an instance of the axiom $!A \to !!A$. Then, by Definition \ref{Def: translation},
\[
\alpha^t = (!A^t \to !!A^t) \wedge ((!!A^t \to D) \to (!A^t \to D)).
\]
For the left conjunct, it is enough to prove $!A \Rightarrow !!A$ by the rule $(R!)$, and the right conjunct has an easy proof by $!A \Rightarrow !!A$. It is clear that the proof has $O(1)$ many lines.

$\bullet$ Rule (mp): Take the following instance of the rule: 
\begin{prooftree}
    \AxiomC{$A$}
    \AxiomC{$A \to B$}
    \BinaryInfC{$B$}
\end{prooftree}
After applying the translation, we should derive $(\, \Rightarrow B^t)$ from $(\, \Rightarrow A^t)$ and $(\, \Rightarrow (A \to B)^t)$ in $iG_D$.
As
\[
(A \to B)^t \Rightarrow A^t \to B^t \qquad A^t, A^t \to B^t \Rightarrow B^t
\]
have $iG_D$-proofs with $O(1)$ many lines, one can use the cut rule to get an $iG_D$-proof of $(\, \Rightarrow B^t)$. 
Note that addressing each occurrence of the modus ponens rule costs $O(1)$ many more lines. 

$\bullet$ Rule $(\operatorname{adj}_u)$: Take the following instance of the rule:
\begin{prooftree}
    \AxiomC{$A$}
    \UnaryInfC{$A \wedge 1$}
\end{prooftree}
By Definition \ref{Def: translation}, $(A \wedge 1)^t= A^t \wedge 1$. As $iG_D \vdash A^t \Rightarrow A^t \wedge 1$, we can prove $(\, \Rightarrow (A \wedge 1)^t)$ from $(\, \Rightarrow A^t)$ in $iG_D$ by adding $O(1)$ many more lines.

$\bullet$ Rule (nec): Take the following instance of the rule:
\begin{prooftree}
    \AxiomC{$A$}
    \UnaryInfC{$!A$}
\end{prooftree}
By Definition \ref{Def: translation}, $(!A)^t=!A^t$. Applying $(R!)$ on $(\, \Rightarrow A)$ proves $\, \Rightarrow (!A)^t$ with adding constant many lines.

This completes the proof of the first part of the theorem. For the second part, it is left to address the weakening axiom, assuming $D=\bot$.

$\bullet$ Axiom $(w)$: Let $\alpha$ be an instance of the axiom $(w)$:
\[
A \to (B \to A).
\]
Then, by Definition \ref{Def: translation}, $\alpha^t= \beta \wedge \gamma$, where
\[
\beta= A^t \to (B \to A)^t \qquad \gamma= (B \to A)^s \to A^s.
\]
We have
\[
\beta= A^t \to \big((B^t \to A^t) \wedge (A^s \to B^s)  \big) \quad \gamma= B^t * A^s \to A^s.
\]
For $(\, \Rightarrow \beta)$, on the one hand, by $(Lw)$, we easily obtain an $iG_D$-proof of $A^t \Rightarrow B^t \to A^t$ with $O(1)$ many lines. On the other hand, since $D=\bot$, by Lemma~\ref{Lem: iGD} the sequent $A^t, A^s \Rightarrow \bot$ has an $iG_D$-proof with $O(|A|) \leq O(|\alpha|)$ many lines. Then, using the axiom $(\bot)$ and the cut rule, we derive an $iG_D$-proof of $A^t \Rightarrow A^s \to B^s$ with $O(|\alpha|)$ many lines. This yields the desired proof for $(\, \Rightarrow \beta)$. 
For $(\, \Rightarrow \gamma)$, using the rules $(Lw)$, $(L*)$, and $(R \to)$, one can easily obtain an $iG_D$-proof with $O(|\alpha|)$ many lines.
\end{proof}



\noindent \textbf{Lemma 32.} (Feasible Deduction Theorem) \emph{Let $\Gamma$ and $\Delta$ be multisets of $\calL_u$-formulas, where $\Delta$ contains at most one formula. For any tree-like $G$-proof $\pi$ of $\Gamma \Rightarrow \Delta$, there exists a multiset $\Sigma_{\pi}$ of single-variable $\calL_u$-formulas and an $\BFL$-proof $\sigma_\pi$ of $\Sigma_\pi, \Gamma \Rightarrow \Delta$, such that $|\Sigma_\pi| \leq O(|\pi|)$, $|\sigma_\pi| \leq O(|\pi|^3)$, and $\LK \vdash \, \Rightarrow \bigwedge \Sigma_\pi^f$.}
\begin{proof}
We construct the multiset $\Sigma_\pi$ and the $\mathbf{FL_e}$-proof $\sigma_\pi$ by induction on the structure of the proof $\pi$ in such a way that every formula in $\Sigma_{\pi}$ has the form $C \wedge 1$ for some formula $C$.

If $\pi$ is an axiom of $\mathbf{FL_e}$, set $\Sigma_{\pi}$ to be the empty set and $\sigma_{\pi} = \pi$. It is clear that this $\Sigma_{\pi}$ and $\sigma_{\pi}$ work.
If $\pi$ is an added initial sequent of the form $A \Rightarrow B$ (resp. $\, \Rightarrow B$) in $G$, set $\Sigma_{\pi}$ to be $\{(A \to B) \wedge 1\}$ (resp. $\{B \wedge 1\}$). Moreover, for $\sigma_{\pi}$ consider the following $\mathbf{FL_e}$-proofs
\begin{center}
\begin{tabular}{c c}
\AxiomC{$A \Rightarrow A$}
\AxiomC{$ \Rightarrow B$}
\BinaryInfC{$A , A \to B \Rightarrow B$}
\UnaryInfC{$A , (A \to B) \wedge 1 \Rightarrow B$}
\DisplayProof \qquad
&
\AxiomC{$B \Rightarrow B$}
\UnaryInfC{$B \wedge 1 \Rightarrow B$}
\DisplayProof
\end{tabular}
\end{center}
respectively. Note that the formulas in $\Sigma_{\pi}$ are of the form:
\begin{center}
$(p \vee \neg p) \wedge 1 \quad (p \to 1)\wedge 1 \quad (\neg p \to 1)\wedge 1 \quad (0 \to p) \wedge 1$
\end{center}
\begin{center}
$(0 \to \neg p) \wedge 1 \quad (0 \to 0 *0) \wedge 1$,
\end{center}
where $p$ is an atomic formula. Therefore, all formulas in $\Sigma_{\pi}$ have the form $C \wedge 1$, and moreover $\LK \vdash \, \Rightarrow \bigwedge \Sigma_{\pi}^f$. In these cases, note that $|\Sigma_{\pi}| \leq O(|\pi|)$ and $|\sigma_{\pi}| \leq O(|\pi|)$.

For the induction step, if the last rule applied in $\pi$ is a unary rule applied to $\pi'$, then set $\Sigma_\pi = \Sigma_{\pi'}$. The proof $\sigma_{\pi}$ is also obtained by applying the same rule to $\sigma_{\pi'}$. For instance, if the last rule in $\pi$ is $(L\wedge)$, then we have:
\begin{center}
\begin{tabular}{l l l}
\AxiomC{$\pi'$}
\noLine
\UnaryInfC{$\Gamma, A \Rightarrow \Delta$}
\UnaryInfC{$\Gamma, A \wedge B \Rightarrow \Delta$}
\DisplayProof
& 
$\rightsquigarrow$
&
\AxiomC{$\sigma_{\pi'}$}
\noLine
\UnaryInfC{$\Sigma_{\pi'}, \Gamma, A \Rightarrow \Delta$}
\UnaryInfC{$\Sigma_{\pi'},\Gamma, A \wedge B \Rightarrow \Delta$}
\DisplayProof
\end{tabular}
\end{center}
Note that 
\[
|\sigma_{\pi}| \leq |\sigma_{\pi'}|+|\Sigma_{\pi'}|+|\pi'| \leq O(|\pi'|^3)+O(|\pi'|) \leq O(|\pi|^3).
\]
If the last rule in $\pi$ is a binary rule applied to $\pi'$ and $\pi''$, then it is either an additive rule (i.e., $(L \vee)$ or $(R \wedge)$) or a multiplicative rule (i.e., $(L \to)$ or $(R *)$). In both cases, we set $\Sigma_\pi = \Sigma_{\pi'} \cup \Sigma_{\pi''}$, and $\sigma_{\pi}$ can be constructed easily from $\sigma_{\pi'}$ and $\sigma_{\pi''}$. For instance, suppose the last rule applied in $\pi$ is the multiplicative rule $(L \to)$:
\begin{center}
\AxiomC{$\pi'$}
\noLine
\UnaryInfC{$\Gamma_1 \Rightarrow A $}
\AxiomC{$\pi''$}
\noLine
\UnaryInfC{$ \Gamma_2, B \Rightarrow \Delta$}
\BinaryInfC{$\Gamma_1, \Gamma_2, A \to B \Rightarrow \Delta$}
\DisplayProof
\end{center}
Then, using the induction hypothesis, we obtain $\sigma_\pi$ as:
\begin{center}
\AxiomC{$\sigma_{\pi'}$}
\noLine
\UnaryInfC{$\Sigma_{\pi'},\Gamma_1 \Rightarrow A $}
\AxiomC{$\sigma_{\pi''}$}
\noLine
\UnaryInfC{$\Sigma_{\pi''}, \Gamma_2, B \Rightarrow \Delta$}
\BinaryInfC{$\Sigma_{\pi'},\Sigma_{\pi''},\Gamma_1, \Gamma_2, A \to B \Rightarrow \Delta$} 
\DisplayProof 
\end{center} 
If the last rule in $\pi$ is the additive rule $(R \wedge)$:
\begin{prooftree}
\AxiomC{$\pi'$}
\noLine
\UnaryInfC{$\Gamma \Rightarrow A $}
\AxiomC{$\pi''$}
\noLine
\UnaryInfC{$ \Gamma \Rightarrow B$}
\BinaryInfC{$\Gamma \Rightarrow A \wedge B$}
\end{prooftree}
Then, by the induction hypothesis we get $\sigma_\pi$ as:
\begin{prooftree}
\AxiomC{$\sigma_{\pi'}$}
\noLine
\UnaryInfC{$\Sigma_{\pi'},\Gamma \Rightarrow A $}
\doubleLine
\UnaryInfC{$\Sigma_{\pi'},\Sigma_{\pi''},\Gamma \Rightarrow A $}
\AxiomC{$\sigma_{\pi''}$}
\noLine
\UnaryInfC{$\Sigma_{\pi''}, \Gamma \Rightarrow B$}
\doubleLine
\UnaryInfC{$\Sigma_{\pi'},\Sigma_{\pi''}, \Gamma \Rightarrow B$}
\BinaryInfC{$\Sigma_{\pi'},\Sigma_{\pi''},\Gamma \Rightarrow A \wedge B$}
\end{prooftree}
where the double lines in the proof mean using the rules $(1w)$ and $(L\wedge)$ to produce $\Sigma_{\pi'},\Sigma_{\pi''}$ from either $\Sigma_{\pi'}$ or $\Sigma_{\pi''}$. Note that every formula in either $\Sigma_{\pi'}$ or $\Sigma_{\pi''}$ is of the form $C \wedge 1$, for a formula $C$. Hence, this form of weakening is possible. Note that 
\[
|\Sigma_{\pi}|= |\Sigma_{\pi'}| + |\Sigma_{\pi''}| \leq O(|\pi'|) + O(|\pi''|) \leq O(|\pi|).
\]
The rightmost inequality follows from the fact that $|\pi'|+|\pi''| \leq |\pi|$, which is a consequence of $\pi$ being tree-like.
For $|\sigma_\pi|$, we have:
\[
|\sigma_\pi| \leq |\sigma_{\pi'}| + |\sigma_{\pi''}| + 2(|\Sigma_{\pi'}| + |\Sigma_{\pi''}|) \cdot (|\Sigma_{\pi'}|+ |\Sigma_{\pi''}|+|\pi'|+|\pi''|)
\]
Since $|\Sigma_{\pi'}| \leq O(|\pi'|)$ and $|\Sigma_{\pi''}| \leq O(|\pi''|)$, we have
\[
2(|\Sigma_{\pi'}| + |\Sigma_{\pi''}|) \cdot (|\Sigma_{\pi'}|+ |\Sigma_{\pi''}|+|\pi'|+|\pi''|) \leq O((|\pi'|+|\pi''|)^2),
\]
which implies
\[
|\sigma_\pi| \leq O(|\pi'|^3) + O(|\pi''|^3) + O((|\pi'|+|\pi''|)^2) \leq O(|\pi|^3).
\] 
The rightmost inequality follows from the fact that $|\pi'|+|\pi''|+1 \leq |\pi|$, which is a consequence of $\pi$ being tree-like.
\end{proof} 

\bibliographystyle{plain}
\bibliography{Ref}

@article{Hyland,
title = {Glueing and orthogonality for models of linear logic},
journal = {Theoretical Computer Science},
volume = {294},
number = {1},
pages = {183-231},
year = {2003},
note = {},
issn = {0304-3975},
doi = {https://doi.org/10.1016/S0304-3975(01)00241-9},
url = {https://www.sciencedirect.com/science/article/pii/S0304397501002419},
author = {Martin Hyland and Andrea Schalk}
}

@article{Raheleh,
title = {Proof complexity of substructural logics},
journal = {Annals of Pure and Applied Logic},
volume = {172},
number = {7},
pages = {102972},
year = {2021},
issn = {0168-0072},
doi = {https://doi.org/10.1016/j.apal.2021.102972},
url = {https://www.sciencedirect.com/science/article/pii/S0168007221000300},
author = {Raheleh Jalali}
}

@article{jerabek,
  title={Substitution {F}rege and extended {F}rege proof systems in non-classical logics},
  author={Je{\v{r}}{\'a}bek, Emil},
  journal={Annals of Pure and Applied Logic},
  volume={159},
  number={1-2},
  pages={1--48},
  year={2009},
  publisher={Elsevier}
}

@article{Cook,
  title={The relative efficiency of propositional proof systems},
  author={Cook, Stephen A and Reckhow, Robert A},
  journal={The journal of symbolic logic},
  volume={44},
  number={1},
  pages={36--50},
  year={1979},
  publisher={Cambridge University Press}
}

@book{Ono,
  title={Residuated lattices: an algebraic glimpse at substructural logics},
  author={Galatos, Nikolaos and Jipsen, Peter and Kowalski, Tomasz and Ono, Hiroakira},
  volume={151},
  year={2007},
  publisher={Elsevier}
}

@article{Avron,
  title={The semantics and proof theory of linear logic},
  author={Avron, Arnon},
  journal={Theoretical Computer Science},
  volume={57},
  number={2-3},
  pages={161--184},
  year={1988},
  publisher={Elsevier}
}

@book{Troelstra,
author = {Troelstra, Anne Sjerp},
address = {Stanford},
booktitle = {Lectures on linear logic},
isbn = {0937073784},
lccn = {lc 91038902},
publisher = {Center for the Study of Language and Information},
series = {CSLI lecture notes; no.29},
title = {Lectures on linear logic },
year = {1992}
}

@article{Girard,
title = {Linear logic},
journal = {Theoretical Computer Science},
volume = {50},
number = {1},
pages = {1-101},
year = {1987},
issn = {0304-3975},
doi = {https://doi.org/10.1016/0304-3975(87)90045-4},
url = {https://www.sciencedirect.com/science/article/pii/0304397587900454},
author = {Jean-Yves Girard}
}

@article{Hrubes,
title = {On lengths of proofs in non-classical logics},
journal = {Annals of Pure and Applied Logic},
volume = {157},
number = {2},
pages = {194-205},
year = {2009},
note = {},
issn = {0168-0072},
doi = {https://doi.org/10.1016/j.apal.2008.09.013},
url = {https://www.sciencedirect.com/science/article/pii/S0168007208001292},
author = {Pavel Hrubeš}
}

@book{Krajicek,
  address = {New York, NY},
  author = {Jan Kraj\'{i}{\v{c}}ek},
  editor = {},
  publisher = {Cambridge University Press},
  title = {Proof Complexity},
  year = {2019},
  series = {Encyclopedia of Mathematics and its Applications},
  volume = {170},
  isbn = {9781108416849}
}

@ARTICLE{Alon,
	author = {Alon, Noga and Boppana, Ravi B.},
	title = {The monotone circuit complexity of boolean functions},
	year = {1987},
	journal = {Combinatorica},
	volume = {7},
	number = {1},
	pages = {1 – 22},
	doi = {10.1007/BF02579196},
	url = {https://www.scopus.com/inward/record.uri?eid=2-s2.0-51249171493&doi=10.1007%2fBF02579196&partnerID=40&md5=d0292f39b6194794bf4ef62649adbc7f},
	type = {Article},
	publication_stage = {Final},
	source = {Scopus}
}

@article{Barr,
  title={*-Autonomous categories and linear logic},
  author={Barr, Michael},
  journal={Mathematical Structures in Computer Science},
  volume={1},
  number={2},
  pages={159--178},
  year={1991},
  publisher={Cambridge University Press}
}

@article{Chu,
  title={*-Autonomous Categories, chapter Constructing *-autonomous categories},
journal={Lecture Notes in
Mathematics},
  author={Chu, Po-Hsiang},
  volume={752, Springer, Berlin, Appendix.},
  year={1979},
  publisher={Springer}
}

@article{Lincoln,
title = {Decision problems for propositional linear logic},
journal = {Annals of Pure and Applied Logic},
volume = {56},
number = {1},
pages = {239-311},
year = {1992},
issn = {0168-0072},
doi = {https://doi.org/10.1016/0168-0072(92)90075-B},
url = {https://www.sciencedirect.com/science/article/pii/016800729290075B},
author = {Patrick Lincoln and John Mitchell and Andre Scedrov and Natarajan Shankar}
}

@article{Horcik,
title = {Disjunction property and complexity of substructural logics},
journal = {Theoretical Computer Science},
volume = {412},
number = {31},
pages = {3992-4006},
year = {2011},
issn = {0304-3975},
doi = {https://doi.org/10.1016/j.tcs.2011.04.004},
url = {https://www.sciencedirect.com/science/article/pii/S0304397511002805},
author = {Rostislav Horčík and Kazushige Terui}
}

@article{AmirProofComp,
  title={Proof Complexity and Feasible Interpolation},
  author={Akbar Tabatabai, Amirhossein},
  journal={arXiv preprint arXiv:2505.03002},
  year={2025}
}

@article{BussPoly,
  title={Polynomial size proofs of the propositional pigeonhole principle},
  author={Buss, Samuel R},
  journal={The Journal of Symbolic Logic},
  volume={52},
  number={4},
  pages={916--927},
  year={1987},
  publisher={Cambridge University Press}
}

@article{Pudlak,
  title={On the complexity of intuitionistic propositional calculus},
  author={Pudl{\'a}k, Pavel},
  journal={Sets and Proofs},
  volume={258},
  pages={197},
  year={1999},
  publisher={Cambridge University Press}
}

@article{PudlakBuss,
  title={On the computational content of intuitionistic propositional proofs},
  author={Buss, Samuel R and Pudl{\'a}k, Pavel},
  journal={Annals of Pure and Applied Logic},
  volume={109},
  number={1-2},
  pages={49--64},
  year={2001},
  publisher={Elsevier}
}

@article{BussMints,
  title={The complexity of the disjunction and existential properties in intuitionistic logic},
  author={Buss, Samuel  R and Mints, Grigori},
  journal={Annals of Pure and Applied Logic},
  volume={99},
  number={1-3},
  pages={93--104},
  year={1999},
  publisher={Elsevier}
}

@inproceedings{Razborov,
  title={Lower bounds on the monotone complexity of some Boolean function},
  author={Razborov, Alexander},
  booktitle={Soviet Math. Dokl.},
  volume={31},
  pages={354--357},
  year={1985}
}

@article{krajivcekFeasible,
  title={Lower bounds to the size of constant-depth propositional proofs},
  author={Kraj{\'\i}{\v{c}}ek, Jan},
  journal={The Journal of Symbolic Logic},
  volume={59},
  number={1},
  pages={73--86},
  year={1994},
  publisher={Cambridge University Press}
}

@article{krajivcekfeasible2,
  title={Interpolation theorems, lower bounds for proof systems, and independence results for bounded arithmetic},
  author={Kraj{\'\i}{\v{c}}ek, Jan},
  journal={The Journal of Symbolic Logic},
  volume={62},
  number={2},
  pages={457--486},
  year={1997},
  publisher={Cambridge University Press}
}

@article{Beyersdorff,
  title={Proof complexity of propositional default logic},
  author={Beyersdorff, Olaf and Meier, Arne and M{\"u}ller, Sebastian and Thomas, Michael and Vollmer, Heribert},
  journal={Archive for Mathematical Logic},
  volume={50},
  number={7},
  pages={727--742},
  year={2011},
  publisher={Springer}
}

@article{Hrubes1,
  title={A lower bound for intuitionistic logic},
  author={Hrube{\v{s}}, Pavel},
  journal={Annals of Pure and Applied Logic},
  volume={146},
  number={1},
  pages={72--90},
  year={2007},
  publisher={Elsevier}
}

@article{Hrubes2,
  title={Lower bounds for modal logics},
  author={Hrube{\v{s}}, Pavel},
  journal={The Journal of Symbolic Logic},
  volume={72},
  number={3},
  pages={941--958},
  year={2007},
  publisher={Cambridge University Press}
}

@article{Haken,
  title={The intractability of resolution},
  author={Haken, Armin},
  journal={Theoretical computer science},
  volume={39},
  pages={297--308},
  year={1985},
  publisher={Elsevier}
}

@article{Null,
  title={Lower bounds on {H}ilbert's Nullstellensatz and propositional proofs},
  author={Beame, Paul and Impagliazzo, Russell and Kraj{\'\i}{\v{c}}ek, Jan and Pitassi, Toniann and Pudl{\'a}k, Pavel},
  journal={Proceedings of the London Mathematical Society},
  volume={3},
  number={1},
  pages={1--26},
  year={1996},
  publisher={Wiley Online Library}
}

@inproceedings{Cutting1,
  title={Lower bounds for cutting planes proofs with small coefficients},
  author={Bonet, Maria and Pitassi, Toniann and Raz, Ran},
  booktitle={Proceedings of the twenty-seventh annual ACM symposium on Theory of computing},
  pages={575--584},
  year={1995}
}

@article{Cutting2,
  title={Lower bounds for resolution and cutting plane proofs and monotone computations},
  author={Pudl{\'a}k, Pavel},
  journal={The Journal of Symbolic Logic},
  volume={62},
  number={3},
  pages={981--998},
  year={1997},
  publisher={Cambridge University Press}
}

@inproceedings{PCalculus1,
  title={Using the {G}roebner basis algorithm to find proofs of unsatisfiability},
  author={Clegg, Matthew and Edmonds, Jeffery and Impagliazzo, Russell},
  booktitle={Proceedings of the twenty-eighth annual ACM symposium on Theory of computing},
  pages={174--183},
  year={1996}
}

@article{PCalculus2,
  title={Lower bounds for the polynomial calculus},
  author={Razborov, Alexander A},
  journal={computational complexity},
  volume={7},
  number={4},
  pages={291--324},
  year={1998},
  publisher={Springer}
}

@article{BDepth1,
  title={The complexity of the pigeonhole principle},
  author={Ajtai, Mikl{\'o}s},
  journal={Combinatorica},
  volume={14},
  number={4},
  pages={417--433},
  year={1994},
  publisher={Springer}
}

@inproceedings{BDepth2,
  title={Exponential lower bounds for the pigeonhole principle},
  author={Beame, Paul and Impagliazzo, Russell and Kraj{\'\i}{\v{c}}ek, Jan and Pitassi, Toniann and Pudl{\'a}k, Pavel and Woods, Alan},
  booktitle={Proceedings of the Twenty-Fourth Annual ACM Symposium on Theory of Computing},
  pages={200--220},
  year={1992}
}

@article{BDepth3,
  title={Exponential lower bounds for the pigeonhole principle},
  author={Pitassi, Toniann and Beame, Paul and Impagliazzo, Russell},
  journal={Computational complexity},
  volume={3},
  number={2},
  pages={97--140},
  year={1993},
  publisher={Springer}
}

@article{BDepth4,
  title={An exponential lower bound to the size of bounded depth Frege proofs of the pigeonhole principle},
  author={Kraj{\'\i}{\v{c}}ek, Jan and Pudl{\'a}k, Pavel and Woods, Alan},
  journal={Random structures \& algorithms},
  volume={7},
  number={1},
  pages={15--39},
  year={1995},
  publisher={Wiley Online Library}
}

@article{Jerabek2,
  title={On the proof complexity of logics of bounded branching},
  author={Je{\v{r}}{\'a}bek, Emil},
  journal={Annals of Pure and Applied Logic},
  volume={174},
  number={1},
  pages={103181},
  year={2023},
  publisher={Elsevier}
}

@article{urquhart,
  title={The complexity of decision procedures in relevance logic {II}},
  author={Urquhart, Alasdair},
  journal={The Journal of Symbolic Logic},
  volume={64},
  number={4},
  pages={1774--1802},
  year={1999},
  publisher={Cambridge University Press}
}

@article{Ferrari,
  title={On the complexity of the disjunction property in intuitionistic and modal logics},
  author={Ferrari, Mauro and Fiorentini, Camillo and Fiorino, Guido},
  journal={ACM Transactions on Computational Logic (TOCL)},
  volume={6},
  number={3},
  pages={519--538},
  year={2005},
  publisher={ACM New York, NY, USA}
}

@article{mints2004,
  title={Intuitionistic Frege systems are polynomially equivalent},
  author={Mints, Grigori and Kojevnikov, Arist Aleksandrovich},
  journal={Zapiski Nauchnyh Seminarov POMI},
  volume={316},
  number={0},
  pages={129--146},
  year={2004},
  publisher={}
}

@article{jevrabek2006,
  title={Frege systems for extensible modal logics},
  author={Je{\v{r}}{\'a}bek, Emil},
  journal={Annals of Pure and Applied Logic},
  volume={142},
  number={1-3},
  pages={366--379},
  year={2006},
  publisher={Elsevier}
}

@article{tabatabai2025,
  title={Universal proof theory: Feasible admissibility in intuitionistic modal logics},
  author={Akbar Tabatabai, Amirhossein and Jalali, Raheleh},
  journal={Annals of Pure and Applied Logic},
  volume={176},
  number={2},
  pages={103526},
  year={2025},
  publisher={Elsevier}
}

@article{lazic2015nonelementary,
  title={Nonelementary complexities for branching {VASS}, {MELL}, and extensions},
  author={Lazi{\'c}, Ranko and Schmitz, Sylvain},
  journal={ACM Transactions on Computational Logic (TOCL)},
  volume={16},
  number={3},
  pages={1--30},
  year={2015},
  publisher={ACM New York, NY, USA}
}

@article{schellinx,
  title={Some syntactical observations on linear logic},
  author={Schellinx, Harold},
  journal={Journal of Logic and Computation},
  volume={1},
  number={4},
  pages={537--559},
  year={1991},
  publisher={Oxford University Press}
}

@article{kanovich2019,
  title={Subexponentials in non-commutative linear logic},
  author={Kanovich, Max and Kuznetsov, Stepan and Nigam, Vivek and Scedrov, Andre},
  journal={Mathematical Structures in Computer Science},
  volume={29},
  number={8},
  pages={1217--1249},
  year={2019},
  publisher={Cambridge University Press}
}

@article{tsinakis2006minimal,
  title={Minimal varieties of involutive residuated lattices},
  author={Tsinakis, Constantine and Wille, Annika M},
  journal={Studia Logica},
  volume={83},
  number={1},
  pages={407--423},
  year={2006},
  publisher={Springer}
}

@book{Quantale,
  title={Quantales and their applications},
  author={Rosenthal, Kimmo I},
  volume={},
  year={1990},
  publisher={Longman Scientific \& Technical}
}

@book{odintsov2008constructive,
  title={Constructive negations and paraconsistency},
  author={Odintsov, Sergei P},
  year={2008},
  publisher={Springer}
}
\end{document}